%% file: READ.tex
\documentclass[twocolumn,3p]{elsarticle}  

\usepackage{fancyhdr}
 \usepackage{amsmath}
\usepackage{cases}
\usepackage{stfloats}
\usepackage{enumerate}
\usepackage{url}
\usepackage{multicol}
\usepackage{graphicx}
\usepackage{wrapfig}
\usepackage{picinpar}
\usepackage{cutwin}
\usepackage{picins}
\usepackage{balance}
\usepackage{algorithm}
\usepackage{algpseudocode}
\usepackage{color}
 \usepackage{etoolbox}

\newtoggle{ACM}
\toggletrue{ACM}
\togglefalse{ACM}

\newtoggle{IEEEcls}
\toggletrue{IEEEcls}
\togglefalse{IEEEcls}

\newtoggle{ElsJ}
\toggletrue{ElsJ}

\floatname{algorithm}{Algorithm}

\usepackage{amsthm}
\theoremstyle{plain}
\newtheorem{theorem}{Theorem}

\theoremstyle{definition}
\newtheorem{definition}{Definition}

\newif\ifNotUse  
 \NotUsetrue
 
 \newif\ifNotUse  
 \NotUsetrue

\input{READ_notation.tex}

\begin{document}
\title{READ: a three-communicating-stage distributed super points detections algorithm} 

\author[jspl,seu_cs]{Jie Xu\corref{cor1}}
\ead{xujieip@163.com}

\cortext[cor1]{Corresponding author}
\address[jspl]{Jiangsu Police Institute}
\address[seu_cs]{School of Computer Science and Engineering, South East University, Nanjing, China}
\address[seu_ns]{School of Cyber Science and Engineering, South East University, Nanjing, China}

\begin{abstract}
A super point is a host that interacts with a far larger number of counterparts in the network over a period of time. Super point detection plays an important role in network research and application. With the increase of network scale, distributed super point detection has become a hot research topic. Compared with single-node super point detection algorithm, the difficulty of super point detection in multi-node distributed environment is how to reduce communication overhead. Therefore, this paper proposes a three-stage communication distributed super point detection algorithm: Rough Estimator based Asynchronous Distributed super point detection algorithm (READ). READ uses a lightweight estimator, the Rough Estimator (RE), which is fast in computation and takes less memory to generate candidate super point. At the same time, the Linear Estimator (LE) is used to accurately estimate the cardinality of each candidate super point, so as to detect the super point correctly. In READ, each node scans IP address pairs asynchronously. When reaching the time window boundary, READ starts three-stage communication to detect the super point. In this paper, we proof that the accuracy of READ in distributed environment is no less than that in the single node environment. Four groups of 10 Gb/s and 40 Gb/s real-world high-speed network traffic are used to test READ. The experimental results show that READ not only has higher accuracy in distributed environment, but also has less than 5$\%$ of communication burden compared with existing algorithms.
\end{abstract} 
 \maketitle
 
\begin{keyword}
super point detection \sep distributed computing \sep network measurement \sep network security
\end{keyword}
\section{Introduction} \label{seq-introduction}
The Internet is one of the most important infrastructures of the modern information society. With the rapid development of China's economy, the bandwidth of core network is increasing year by year. According to the latest statistics of China Internet Information Center (CNNIC), as of December 2018, China's international export bandwidth has reached 8,946,570 Mbps, with an annual growth rate of 22.2$\%$\cite{report:cnni:chinesenetreport:en}. It is a worldwide problem to manage such a large-scale network effectively and ensure its safe operation.

In the face of complex network environment, the monitoring and protection of backbone network is the most important and basic step\cite{thesis:zap:seu:2015:en}. Internet management under the condition of large data-level network traffic is a hot research subject, which can be carried out from different aspects at the industrial and academic levels. To pay more attention to some core hosts in the network is a way to improve the efficiency of network management\cite{ieeec2018:generalidsaccelerationforhighspeednetworks}.

The super point in the Internet is such a kind of core host\cite{hsd:streamingalgorithmfastdetectionsuperspreaders}. It is generally believed that a super point refers to a host that communicate with lots of other hosts. Super points play important roles in the network, such as servers, proxies, scanners\cite{instr:asurveyintrusiondetectiontechniquesincloud:chiragmodi}, hosts attacked by DDoS, etc. The detection and measurement of super points are important to network security and network management\cite{hsd:infcom:simpleadaptiveidentificationsuperspreadersflowsampling}.

With the increase of network size, large-scale networks usually contain multiple border entries and exits. How to detect the super point from multiple nodes is a new requirement for super point detection. Some existing algorithms, such as DCDS\cite{hsd:adatastreamingmethodmonitorhostconnectiondegreehighspeed}, VBFA\cite{hsd:detectionsuperpointsvectorbloomfilter} and CSE\cite{hsd:compactspreadestimatorsmallhighspeedmemory} and so on, can realize distributed super point detection by adding data merging process. However, in the distributed environment, DCDS, VBFA, CSE must send all the whole used memory, which is more than 300MB for a 10Gb/s network, to the main server. When detecting the super point, such a large data transmission between the sub-node and the global server will cause the peak traffic of network communication and increase the communication delay. How to reduce the communication overhead in distributed environment is a difficult problem in the research of distributed super point detection.

Super points account for only a small portion of all hosts. In theory, only the data related to the super point should be sent to the global server to complete the super point detection. Based on this idea, a distributed super point detection algorithm, asynchronous distributed algorithm based on rough estimator (READ), is proposed in this paper. READ uses a lightweight rough estimator (RE) to generate candidate super points. Because RE takes up less memory, each sub-node only needs to send a small amount of data to the global server to generate candidate super points. READ not only reduces the detection error rate, but also reduces the communication overhead by transferring data related to candidate super points to the global server. The main contributions of this paper are as follows:
\begin{itemize}
\item A method of generating candidate super point in distributed environment using lightweight estimators is proposed.
\item A distributed super point detection algorithm READ with low communication overhead is proposed.
\item It is proved theoretically that READ has lower error rate in distributed environment.
\item Using the real-world high-speed network traffic to evaluate the performance of READ.
\end{itemize}

In section \ref{sec-relatedWork}, we introduce the rough estimator and the linear estimator for estimating host’s cardinality, as well as the existing algorithms for super point detection. Section \ref{sec-DistributedSP-detecte-model} discuss the model and difficulty of distributed super point detection. Section \ref{sec-READalgorithm} introduces the operation principle of READ, and how READ works. Section \ref{sec-READunderSlidingWindow} introduces how to modify READ to work under sliding time window. Section \ref{sec-experiments} shows the experiment of READ with 10 Gb/s and 40 Gb/s real world network traffic, and analyses the detection accuracy of READ in distributed environment and the communication overhead between sub-nodes and the global server. Section \ref{sec-conclusion} summarizes READ.

\section{Related work} \label{sec-relatedWork}
Super point detection is a hotspot in the field of network research and management. For the sake of narrative convenience, this section first gives relevant definitions.
\subsection{Related definitions}
All of the super point detection algorithms are based on network traffic and belong to passive network measurement. The original data used in the algorithm is the IP address collected from the network. For network managers, the measuring place is usually located at the boundary of the managed network, as shown in Figure \ref{fig_ObservationNode_model}. The host in $\ANetmath$ communicates with those hosts in $\BNetmath$ through the boundary router. IP address pairs such as $<\aipmath, \bipmath>$ can be extracted from each packet passing through the border router, where $\aipmath \in \ANetmath$, $\bipmath \in \BNetmath$. For the host $\aipmath$ in $\ANetmath$, its cardinality is defined as follows:
\begin{definition}[Opposite host set / cardinality]
\label{def-oppositehostsAndCardinaity}
In time window $\TWinmath$, for a host $\aipmath \in \ANetmath$, the set of all hosts in $\BNetmath$ that communicating with it is called the opposite host set of $\aipmath$, and is denoted as $\OPmath{\aipmath}{\TWinmath}$. The size of $\OPmath{\aipmath}{\TWinmath}$ is called the cardinality of $\aipmath$, which is denoted as $|\OPmath{\aipmath}{\TWinmath}|$.
\end{definition}

The cardinality is one of the important network attribute\cite{trafficclas2015j:aclassorientedfeatureselection}, and it is also the only criteria of super point judgement.

\begin{definition}[Super point]
\label{def-superPoint}
In the time window $\TWinmath$, the host whose cardinality exceeds the specified threshold θ is called a super point.
\end{definition}

In this paper, without losing generality, it is assumed that the super point detection is only for $\ANetmath$. Threshold θ is set by the users according to different situations, such as detecting DDoS attacks, locating servers and so on.

Cardinality estimation is the basis of super point detection. In the next section, we will introduce the commonly used algorithm for cardinality estimating in super point detection.

\begin{figure}[!ht]
\centering
\includegraphics[width=0.47\textwidth]{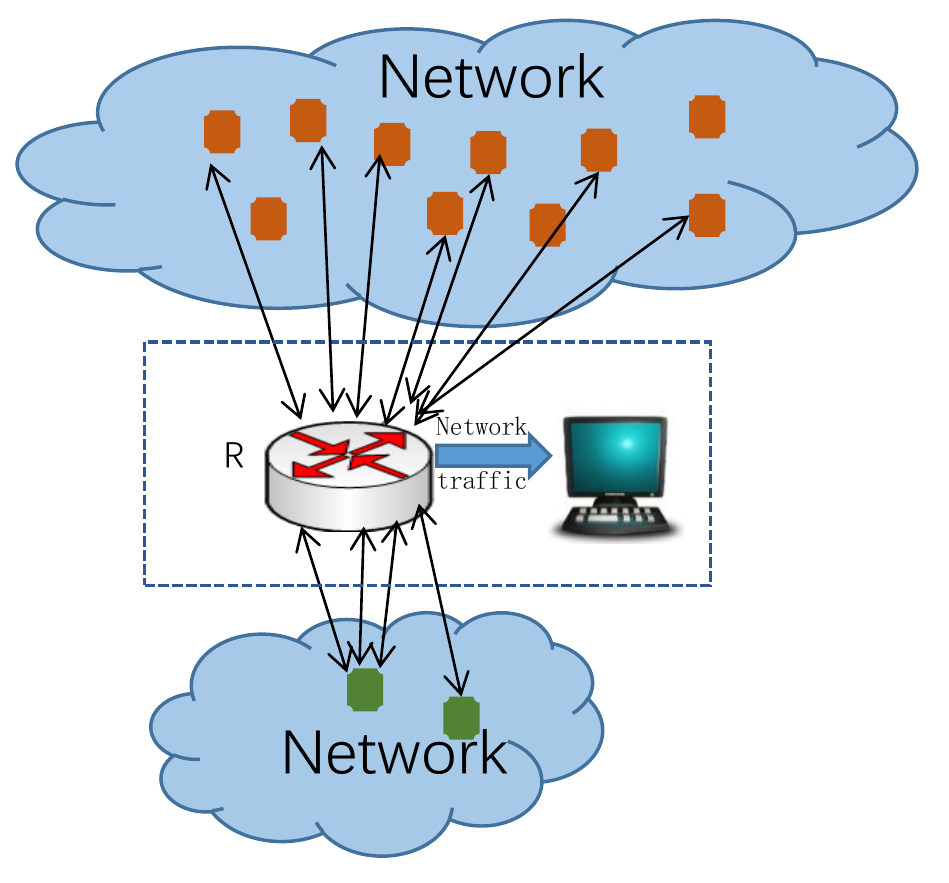}
\caption{The observation node on network boarder}
\label{fig_ObservationNode_model}
\end{figure}

\subsection{Cardinality Estimation}
Cardinality is an important attribute in network research\cite{tranieee2014:towardsmoreefficientcardinalityestimationforlarge-scalerfidsystems}. At the same time, the calculation of cardinality is also the basis of super point detection\cite{tranieee2013:contentionbasedestimationneighborcardinality}. Therefore, this sub section introduces the algorithm of host’s cardinality estimating\cite{confieee2015:towardsconstanttimecardinalityestimationforlargescalerfidsystems}.

There are many cardinality estimating algorithms, such as PCSA algorithm\cite{dc1983:probabilisticcounting}, HyperLogLog algorithm\cite{dc:hyperloglogtheanalysisofanearoptimalcardinalityestimationalgorithm}, Linear Estimator (LE) algorithm\cite{dc:alineartimeprobabilisticcountingdatabaseapp} and so on. LE algorithm is widely used in super point detection because of its high accuracy and simple operation.

Let $\LEbitsSetmath$ denote a set of bits and $|\LEbitsSetmath|$ denote the number of bits in $\LEbitsSetmath$. LE uses $\LEbitsSetmath$ to record and estimate the opposite hosts of $\aipmath$. Each bit in $\LEbitsSetmath$ is initially set to zero. For any opposite host $\bipmath$, LE maps it to a bit in $\LEbitsSetmath$ by using the hash function $\LEophostHashmath{\bipmath}$ and sets the bit to 1. At the end of time window $\TWinmath$, LE uses the following formula to estimate $|\OPmath{\aipmath}{\TWinmath}|$. Where $n_0$ denotes the number of bits in C with value of 0.

\begin{equation}
\label{eq-LE-cardinalityEstimation}
|\OPmath{\aipmath}{\TWinmath}|'=-|\LEbitsSetmath|*log(n_0/(|\LEbitsSetmath|))
\end{equation}

The estimation error of LE is related to $|\OPmath{\aipmath}{\TWinmath}|$ and the number of counter $|\LEbitsSetmath|$. Define the ratio of $|\OPmath{\aipmath}{\TWinmath}|$ to $|\LEbitsSetmath|$ as a load factor, marked $\LoadFactormath$. The estimated standard deviation of LE is $\sqrt{\frac{(e^\LoadFactormath-\LoadFactormath-1)}{\LEbitsSetmath}}$.

When $|\OPmath{\aipmath}{\TWinmath}|$ is determined, the larger $|\LEbitsSetmath|$ is, the higher the estimation accuracy of LE is. However, the larger $|\LEbitsSetmath|$, the more memory space LE occupies, and the longer time it takes to estimate the cardinality.

In order to compensate for the deficiency of LE, Jie et al.\cite{ispa2017:highspeednetworksuperpointsdetectionbasedslidingwindowgpu} proposed a lightweight rough estimator (RE). RE only takes 8 bits to determine whether $\aipmath$ is a candidate super point. At initialization, RE sets all 8 bits to 0. For each opposite host $\bipmath$ of $\aipmath$, RE maps $\bipmath$ to a random integer $\bipHashedValuemath$ between 0 and $2^{32}-1$ using hash function $\REophostrandHashmath{\bipmath}$, and then compares the lowest significant bit of $\bipHashedValuemath$ with a real number $\tau$.  The lowest significant bit is the position of the first bit ``1" starting from the right. For example, the binary formatter of integer 200 is ``11001000", its lowest significant bit is 3. Let $\LSBmath{x}$ denote the lowest significant bit of integer x. $\tau$ is used to determine whether update a bit. The definition of $\tau$ is as follows.

\begin{equation}
\label{eq-RE_tau_get}
\tau=log2(\theta / g)
\end{equation}

If $\LSBmath{\bipHashedValuemath}\geq \tau$, RE maps $\bipmath$ to one of 8 bits using a hash function and sets the bit to 1. Denote this hash function as $\REophostbitHashmath{\bipmath}$. When the number of bits with a value of 1 is greater than or equal to 3, RE determines $\bipmath$ as a candidate super point. As a lightweight estimator, RE can quickly determine candidate super point, but it cannot accurately estimate the cardinality. Jie et al.\cite{hpcc2018:srla:arealtimeslidingtimewindowsuperpointcardinalityestimationalgorithmforhighspeednetworkongpu} used RE as a preliminary screening tool to reduce the range of candidate super points, and combined with LE to realize real-time detection of super points under sliding time window. A detailed analysis of RE could be found in \cite{IEEEAccess2019:SRLA}.

\subsection{Super point detection}
From the introduction in the previous sub section, LE and RE can estimate the cardinality of a host and determine whether a host is a candidate super point. However, there are a large number of active IP\cite{j2017:unwisdomcrowdsaccuratelyspottingmaliciousipclustersusinnotsoaccurateipblacklists} in the actual network. At the beginning of the time window, it is not known which IP will become a super point. The task of the super point detection algorithm is to detect the super points from these IP based on the cardinality estimation algorithm. In this paper, we call the memory that used to record the opposite hosts’ information as master data structure.

A simple and straightforward method of super point detection is to record each host $\aipmath$ and its opposite IP. But this is unrealistic, because there are a lot of IP addresses in high-speed networks. Accurately recording each IP and its opposite hosts not only requires a lot of memory, but also a lot of memory access times\cite{j2012:anospfintegratedroutingstrategyqosawareenergysavingipbackbonenetworks}. Therefore, the estimation-based super point detection algorithms using fixed amount of memory have attracted wide attention, and a large number of super points detection algorithms have emerged, such as CBF\cite{hsd:linespeedaccuratesuperspreaderidentificationdynamicerrorcompensation}, DCDS\cite{hsd:adatastreamingmethodmonitorhostconnectiondegreehighspeed}, VBFA\cite{hsd:detectionsuperpointsvectorbloomfilter} and CSE\cite{hsd:compactspreadestimatorsmallhighspeedmemory}.

CBF\cite{hsd:linespeedaccuratesuperspreaderidentificationdynamicerrorcompensation} is a super point detection algorithm based on the principle of Bloom filter. It uses Bloom filter to remove duplicate IP address pairs, and uses a data structure derived from Bloom filter, called Counting Bloom filter, to record opposite IP information. The algorithm uses Bloom filter to avoid multiple updates of the master data structure by the same IP address pair, and improves the speed of the algorithm. When updating the counting Bloom filter, only increment some counters with 1, and no other complicated calculation is needed. Since each counter can be used by multiple hosts, the memory usage of the algorithm is low. Although Bloom filter can avoid multiple updates of CBF to an IP address pair, it may also cause omissions of some IP address pairs. In distributed environment, an IP address pair will appear on different nodes, which will be updated by different nodes many times. Therefore, CBF can not be applied to distributed environment.

DCDS\cite{hsd:adatastreamingmethodmonitorhostconnectiondegreehighspeed}, VBFA\cite{hsd:detectionsuperpointsvectorbloomfilter} and CSE\cite{hsd:compactspreadestimatorsmallhighspeedmemory} all use LE to estimate host’s cardinality. DCDS\cite{hsd:adatastreamingmethodmonitorhostconnectiondegreehighspeed} uses China Remainder Theorem (CRT)\cite{jnlsd2015:anewrobustchineseremaindertheoremwithimprovedperformanceinfrequencyestimationfromundersampledwaveforms} to restore candidate super point. However, when mapping $\aipmath$ to LE, DCDS needs to use CRT principle, which takes up more computing time and is not conducive to the improvement of algorithm speed. VBFA does not use computationally complex CRT to recover candidate super points, but maps $\aipmath$ to different LE according to the principle of  Bloom filter\cite{bf:anewanalysisoffalsepositiverate}. The length of LE array used to recover candidate super points in VBFA is fixed. As the number of host increases, each LE is used to estimate too many hosts’ cardinalities. At this time, the number of hot LE (whose cardinality is bigger than threshold) in LE array increases correspondingly. The number of hot LEs that need to be tested also increases, which increases the time to recover candidate super points. CSE uses virtual LE to estimate the number of counterparts. CSE assigns a virtual LE to each $\aipmath$. Each bit in virtual LE associates with a physical bit in the bit pool. CSE achieves bit-level sharing and makes more efficient use of memory. Each $\aipmath$ associates with only one virtual estimator, so only one physical bit needs to be updated when scanning each IP address pair, and memory access times are less than DCDS and VBFA. CSE cannot generate candidate super points after scanning all IP address pairs in a time window like DCDS and VBFA. Therefore, CSE saves all hosts in $\ANetmath$ as candidate super points, when scanning IP address pairs. It increases the number of candidate super points and the time used to estimate the cardinalities of candidate super points.

DCDS, VBFA and CSE can run in distributed environment. In distributed environment, DCDS and VBFA collect LE from all nodes, and merge these LE sets according to ``bit or" mode; CSE collects bit pools from all nodes, and merges these bit pools according to ``bit or" mode. Then the super points are detected according to the unioned LE set or bit pool. Although DCDS, VBFA and CSE can run in a distributed environment, they need to collect all LE or bit pools from each distributed node, which leads to low communication efficiency. This paper presents an algorithm that can realize distributed super point detection by collecting only fraction of LE sets, which reduces the communication in distributed environment. 
\subsection{Notations and symbols}
To facilitate reading, Table \ref{tbl_notationTbl} lists some commonly used symbols and abbreviations in this article. In Table \ref{tbl_notationTbl}, RE cube, RE array, LE array are data structures used in our algorithm, and they will be described in detail in section \ref{sec-READalgorithm}.

\input{READ_table_of_notation_content.tex}

\section{Distributed super point detection model and difficulty} \label{sec-DistributedSP-detecte-model}
A network connected to the Internet may have multiple border routers, as shown in Figure \ref{fig_ObservationNodeDistributed_model}. For example, a campus network access to multiple Internet Service Provider(ISP). Assuming that there is an observation node at each border router. Traffic can be observed and analyzed independently on each node. This section will discuss the algorithm of super point detection in distributed environment.
\begin{figure}[!ht]
\centering
\includegraphics[width=0.47\textwidth]{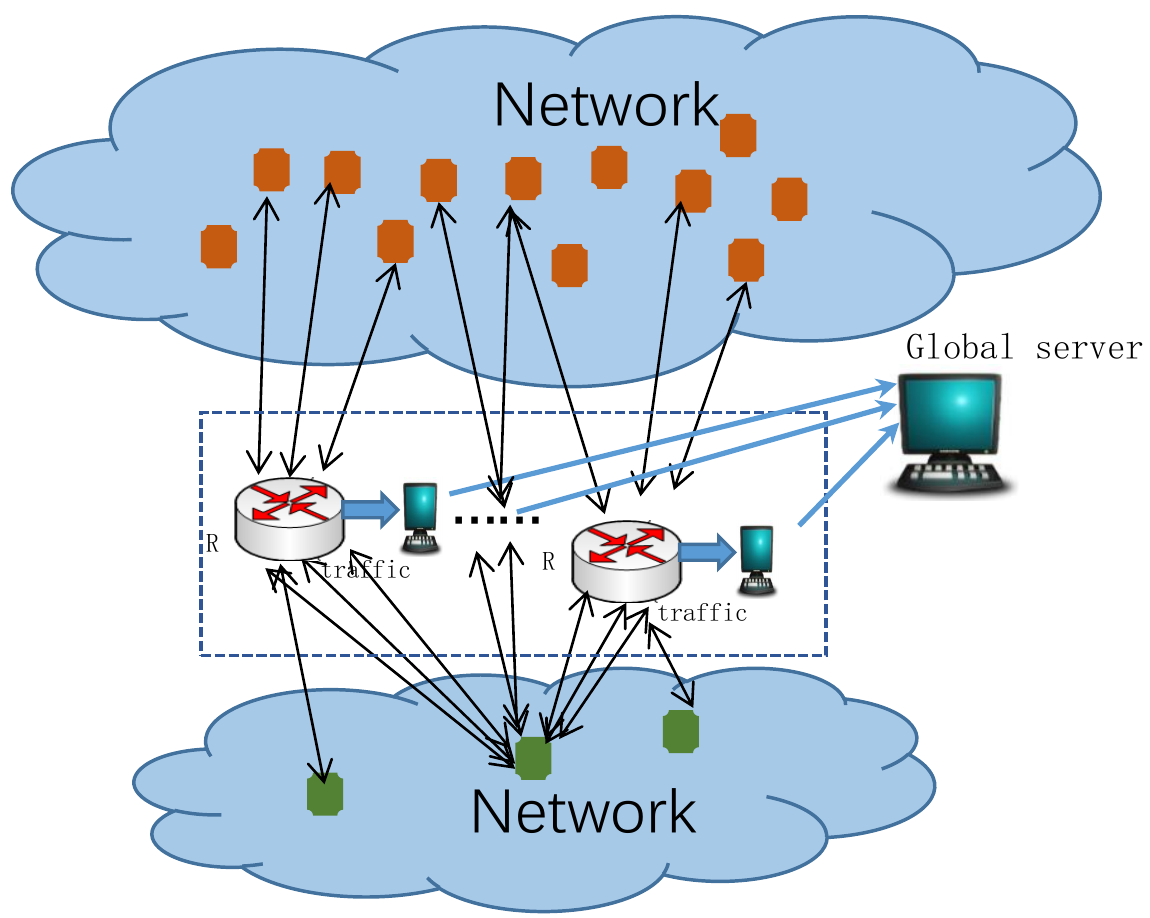}
\caption{ Super point detection in distributed environmentr}
\label{fig_ObservationNodeDistributed_model}
\end{figure}
\subsection{Detection model}
For a host $\aipmath$ in the network, it may interact with different opposite hosts through different border routers. At this time, only part of the traffic of $\aipmath$ can be observed at each observation node. Assuming that the host $\aipmath$ communicates with other networks in the Internet through $\SubnodeNummath$ border routers, only part of the traffic of $\aipmath$ is forwarded on each border router. At this time, the cardinality of $\aipmath$ observed at each border router may be less than the threshold, but the cardinality of $\aipmath$ observed from all observation nodes is larger than the threshold, which will lead to the omission of super points. Therefore, it is a meaningful work to detect the super point in distributed environment.

In the distributed environment, the global server collects data from all observation nodes and performs super point detection. The research of super point detection in distributed environment is to study which data the global server collects from the observation nodes and how to detect the global super points on the global server.
\subsection{Requirements and difficulties}
In order to avoid missing super point in distributed environment, it is necessary to detect them globally. A simple method is to send the IP address pairs extracted from each observation node to a global server that processes all data, and then detect the super point on global server. This method needs to transfer a large amount of data between the global server and observation nodes. Therefore, the method of sending all IP addresses to the global server and detecting the super point on the global server cannot process the high-speed network data in real time because of the long communication time.

Another method of super point detection in distributed environment is to run super point detection algorithms, such as DCDS, VBFA and CSE, at each observation node and then send only the master data structure to the global server for super point detection. Compared with the method of transferring all IP addresses to the global server, the method of transferring only the master data structure to the global server reduces the communication overhead between observation nodes and the global server.

But when using this method, all observation nodes need to transmit the master data structure to the global server. When the number of observation nodes increases, the total amount of data transferred between all observation nodes and the global server will also increase. Moreover, the size of the master data structure is related to the error rate of the algorithm: the larger the master data structure, the lower the error rate of the algorithm. Therefore, the communication overhead between the observation node and the master node cannot be reduced by reducing the size of the master data structure. In addition, the transmission of all master data structures will generate a large amount of burst traffic at the end of the time window, which will increase the network burden.

How to avoid sending all master data structures to the global server and reduce the communication between observation nodes and the global server is a difficult problem in distributed environment.
\subsection{Solution of this paper}
If only part of the cardinality estimation structure at the observation node is sufficient to detect the global super point, then there is no need to transfer all of them between the observation node and the global server, which can further reduce the communication overhead. Based on this idea, this paper proposes a low communication overhead distributed super point detection algorithm: Rough Estimator based Asynchronous Distributing Algorithm (READ).

In distributed environment, it is necessary to recover the global candidate super points at the end of the time window according to the information recorded at all observation nodes. DCDS and VBFA have the function of recovering candidate super points. But DCDS and VBFA have to use LE to recover candidate super points. Although LE has a high accuracy, it also occupies a high amount of memory, resulting in a large amount of communication between observation nodes and the global server.

RE not only runs fast, but also occupies less memory. If RE is used to generate candidate super points, a small amount of memory can be used to generate global candidate super points. The global server collects LE related to candidate super points from all observation nodes for estimating the cardinalities of candidate super points, and then completes super points detection without transmitting all cardinality estimation structure.
In the next section, we will describe how READ works.
\section{RE based distributed super points detection algorithm READ} \label{sec-READalgorithm}
In this section, we will introduce our low communication overhead distributed super points detection algorithm Rough Estimator based Asynchronous Distributed super points detection algorithm(READ).
\subsection{Principle of READ}
READ uses a data structure that can recover candidate super points to achieve distributed super points detection. It uses RE to recover candidate super points and LE to estimate cardinality of each candidate super point. Therefore, the master data structure of READ includes two parts: RE set and LE set. Scanning IP address pairs and estimating cardinalities are operations on RE and LE sets. REDA algorithm contains three main steps:
\begin{itemize}
\item Scan IP pair on each observation node. Each observation node scans each IP address pair passing through it and updates the RE and LE sets on it.
\item Generate candidate super points in global server. The global server collects RE sets from all observation nodes, merges these RE sets, and generates candidate super points using the merged RE sets.
\item Estimate cardinalities and filter super points. After the candidate super points are obtained, the global server collects LE related to each candidate super point from all observation nodes, and estimates the cardinalities of candidate super points based on these LE. 
\end{itemize}

According to the above analysis, in READ, the communication between observation nodes and the global server is divided into three stages:
\begin{itemize}
\item Each observation node sends RE set to global server;
\item The global server distributes candidate super points to each observation node;
\item Each observation node sends LE of every candidate super point to the global server;
\end{itemize}

For READ, the sum of the communication in the three stages above is the total communication between an observation node and the global server in a time window. The number of LEs sent by observation nodes to the global server equals to the number of candidate super points. Since the number of candidate super points is less than the number of LE in the master data structure, the amount of data sent by each observation node to the global server is less than the size of LE set.

\subsection{Scanning IP pair in distributed environment}
Distributed scanning IP address pairs is to scan the IP address pairs collected at each observation node. Let $\SubnodeObjmath{\SubnodeIdxmath}$ denote the $\SubnodeIdxmath$-th observation node and \IPair{\TWinmath,\SubnodeIdxmath} denote all IP address pairs in time window $\TWinmath$ on $\SubnodeObjmath{\SubnodeIdxmath}$. READ uses RE estimator and LE estimator to record IP information. Each observation node has the same cardinality estimation structure: the same number of RE and LE, and the same number of counters in RE and LE. The basic operation of $\SubnodeObjmath{\SubnodeIdxmath}$ when scanning IP address pairs is to update RE and LE.

READ uses RE to generate global candidate super points, and LE to estimate the cardinality of each global candidate super point. In distributed environment, because only part of the network traffic can be observed at each observation node, it is impossible to determine whether a host is a global candidate super point according to RE when scanning IP address pairs. In distributed environment, the algorithm of super point detection must be able to recover the global candidate super points directly, such as DCDS and VBFA.

In order to recover candidate super points, READ adopts a new data structure, Rough Estimator Cube (REC). REC is a three-dimensional data structure composed of RE, as shown in Figure \ref{fig_RE_cube_structure}.

The basic element of REC is RE. Several RE constitutes a one-dimensional RE vector (REV); the set of REV constitutes a two-dimensional RE array (RE Array, REA). The three-dimensional REC can be regarded as a set of REA, which contains $2^ {\RECaryNlgtmath}$ REA and $\RECaryNlgtmath$ is a positive integer less than 32. Each REA of REC has the same structure, that is, the REA contains the same number of REV, and the associating REV contains the same number of RE. Let $\RECrowNlgtmath$ denote the number of REV contained in REA and $2^{\RECcolNlgtmath{i}}$ denote the number of RE contained in the $i$th REV. Three indexes can be used to locate a RE in REC accurately.

All observation nodes have their own REC, and the structure of REC at different observation nodes is the same, that is, the $\RECaryNlgtmath$, $\RECrowNlgtmath$ , $2^{\RECcolNlgtmath{i}}$ of REC at different observation nodes are the same. When the IP address pair is scanned at the observation node, the REC at the observation node will be updated. Let $\RECobjmath{\SubnodeIdxmath}$ denote the REC on the observation node \SubnodeObj{\SubnodeIdxmath}, $\RECobjonlymath_{(i,j,k)}^\SubnodeIdxmath$ denote the $j$-th RE of the $i$-th REV on the $k$-th REA, where k is an integer between 0 and $2^{\RECaryNlgtmath}-1$, i is an integer between 0 and $\RECrowNlgtmath$ -1, and j is an integer between 0 and $2^{\RECcolNlgtmath{i}}$-1. 
In time window $\TWinmath$, for each IP address pair $<\aipmath, \bipmath >$ of \IPair{\TWinmath,\SubnodeIdxmath}, READ selects $\RECrowNlgtmath$ RE from \RECobj{\SubnodeIdxmath} according to $\aipmath$, and updates $\RECrowNlgtmath$ RE with $\bipmath$. How to map $\aipmath$ to $\RECrowNlgtmath$ RE in REC determines whether READ can recover global candidate super points from REC.

\begin{figure}[!ht]
\centering
\includegraphics[width=0.47\textwidth]{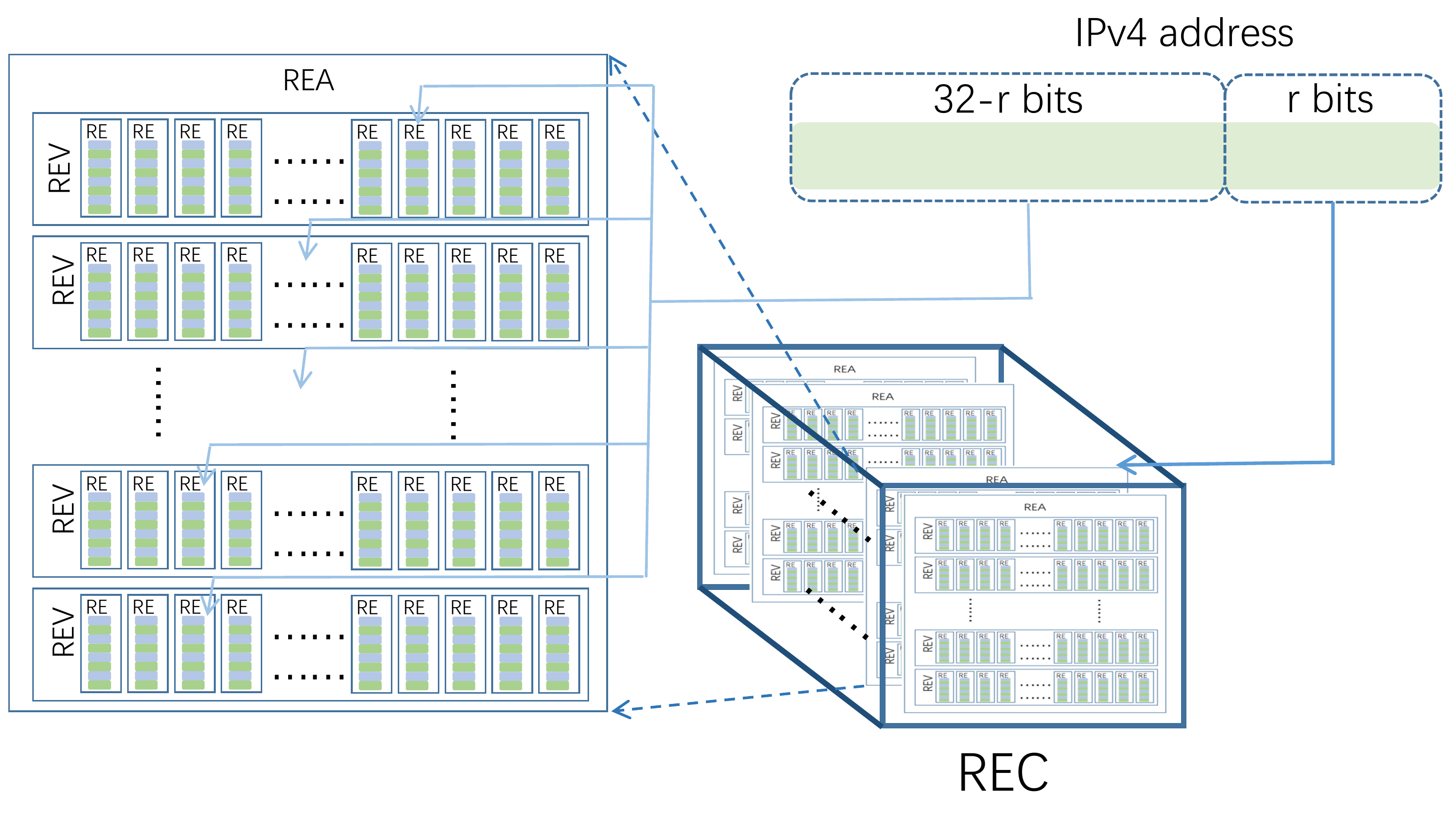}
\caption{Structure of RE cube}
\label{fig_RE_cube_structure}
\end{figure}

The $\RECrowNlgtmath$ RE associating with $\aipmath$ are located in the same REA. READ divides $\ANetmath$ into two parts: the first part is $\RECaryNlgtmath$ bits on the right (Right Part, RP), and the second part is 32-$\RECaryNlgtmath$ bits on the left (Left Part, LP).

READ selects a REA in the REC based on the IP of $\aipmath$. REC has $2^{\RECaryNlgtmath}$ REA, so the RP of $\aipmath$ can determine only one REA in the REC. READ divides $\ANetmath$ into $2^{\RECaryNlgtmath}$ subsets according to $\RECaryNlgtmath$ bits on the right side of the IP address. Each subset of $\ANetmath$ associates with a REA in the REC. During the operation of the algorithm, the number of RE in the REC is fixed, and each RE is used to record opposite hosts of multiple $\aipmath$. When $\ANetmath$ contains many IP addresses, by increasing $\RECaryNlgtmath$, the number of hosts sharing the same RE can be reduced.

The LP of $\aipmath$ is used to select $\RECrowNlgtmath$ RE in REA, i.e. one RE from each REV. Let \REidxINithRow{\aipmath}{i} denote the index of RE in the $i$-th REV, $0 \leq \REidxINithRowmath{\aipmath}{i} \leq 2^{\RECcolNlgtmath{i}}-1$. \REidxINithRow{\aipmath}{i} is an integer containing \RECcolNlgt{i} bits. Let \REidxINithRow{\aipmath}{i} [j] denote the $j$-th bit in \REidxINithRow{\aipmath}{i}, $0 \leq j \leq \RECcolNlgtmath{i}$. READ selects \RECcolNlgt{i} bits from the LP of $\aipmath$ as the value of \REidxINithRow{\aipmath}{i}. Let \IPLeftPart{\aipmath} denote the LP of $\aipmath$, \IPLeftPart{\aipmath}[i] denote the $i$-th bit of \IPLeftPart{\aipmath}, $0 \leq i \leq 32-$\RECaryNlgtmath$-1$. Each bit in \REidxINithRow{\aipmath}{i} associates with a bit in \IPLeftPart{\aipmath}, as shown in Figure \ref{fig_getREidx_from_LP}.

\begin{figure}[!ht]
\centering
\includegraphics[width=0.47\textwidth]{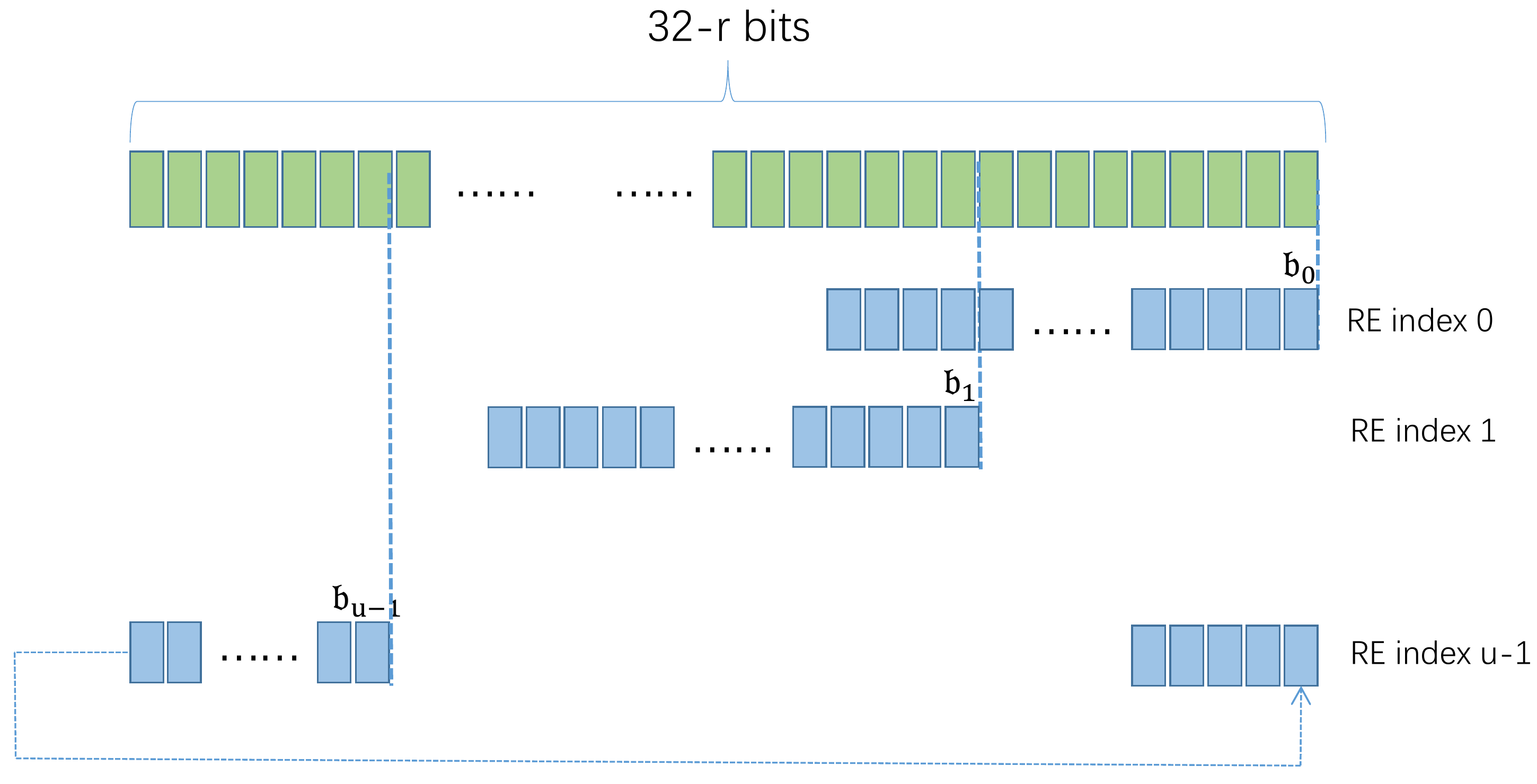}
\caption{Locate RE by left part of IP address}
\label{fig_getREidx_from_LP}
\end{figure}

When selecting bits from \IPLeftPart{\aipmath} as \REidxINithRow{\aipmath}{i}, READ first determines which bit in \IPLeftPart{\aipmath} is \REidxINithRow{\aipmath}{i} [0], and then calculates the other bits in \REidxINithRow{\aipmath}{i}. Let \REidxStartBit{i} denote the index of the 0th bit of \REidxINithRow{\aipmath}{i} in \IPLeftPart{\aipmath}, i.e. \REidxINithRow{\aipmath}{i} [0]=\IPLeftPart{\aipmath} [\REidxStartBit{i}]. Each bit of \REidxINithRow{\aipmath}{i} is calculated according to the following formula:

\begin{equation}
\label{eq-READgetREidxFromLPofIP}
\REidxINithRowmath{\aipmath}{i}[j]=\IPLeftPartmath{\aipmath} [(\REidxStartBitmath{i}+j)mod(32-\RECaryNlgtmath)],0\leq j \leq \RECcolNlgtmath{i}-1
\end{equation}

\REidxStartBit{i}（$0≤i≤ \RECrowNlgtmath-1$）is a parameter of READ, which is determined at the beginning of the algorithm. In order to recover the global candidate super point from REC, \REidxStartBit{i} meets the following conditions when setting:
\begin{itemize}
\item $\REidxStartBitmath{0}=0$
\item $\REidxStartBitmath{i}<\REidxStartBitmath{i+1}<31-\RECaryNlgtmath,i\in [0,\RECrowNlgtmath-2]$
\item $\REidxStartBitmath{i+1}<\REidxStartBitmath{i}+\RECcolNlgtmath{i}-1,i∈[0,\RECrowNlgtmath-2]$
\item $\REidxStartBitmath{\RECrowNlgtmath-1}+\RECcolNlgtmath{\RECrowNlgtmath-1}>31-\RECaryNlgtmath$
\end{itemize}

The above conditions ensure that each bit in \IPLeftPart{\aipmath} appears in at least one \REidxINithRow{\aipmath}{i}, and that there are the same bits between two adjacent \REidxINithRow{\aipmath}{i} (associating with the same bit in \IPLeftPart{\aipmath}). When restoring global candidate super points, READ extracts the associating bits of \IPLeftPart{\aipmath} from all \REidxINithRow{\aipmath}{i} to recover \IPLeftPart{\aipmath}, and reduces the number of global candidate super points by using the repeated bits between two adjacent \REidxINithRow{\aipmath}{i}.

RE estimator only determine whether the host is a global candidate super point, but cannot give an estimate of the cardinality. Therefore, READ uses LE to estimate the cardinality of each global candidate super points.

READ uses LE array of $\LEArowNmath$ rows and $\LEAcolNmath$ columns to record the opposite hosts of $\aipmath$, as shown in Figure \ref{fig_StructureOfLEArray}.

\begin{figure}[!ht]
\centering
\includegraphics[width=0.47\textwidth]{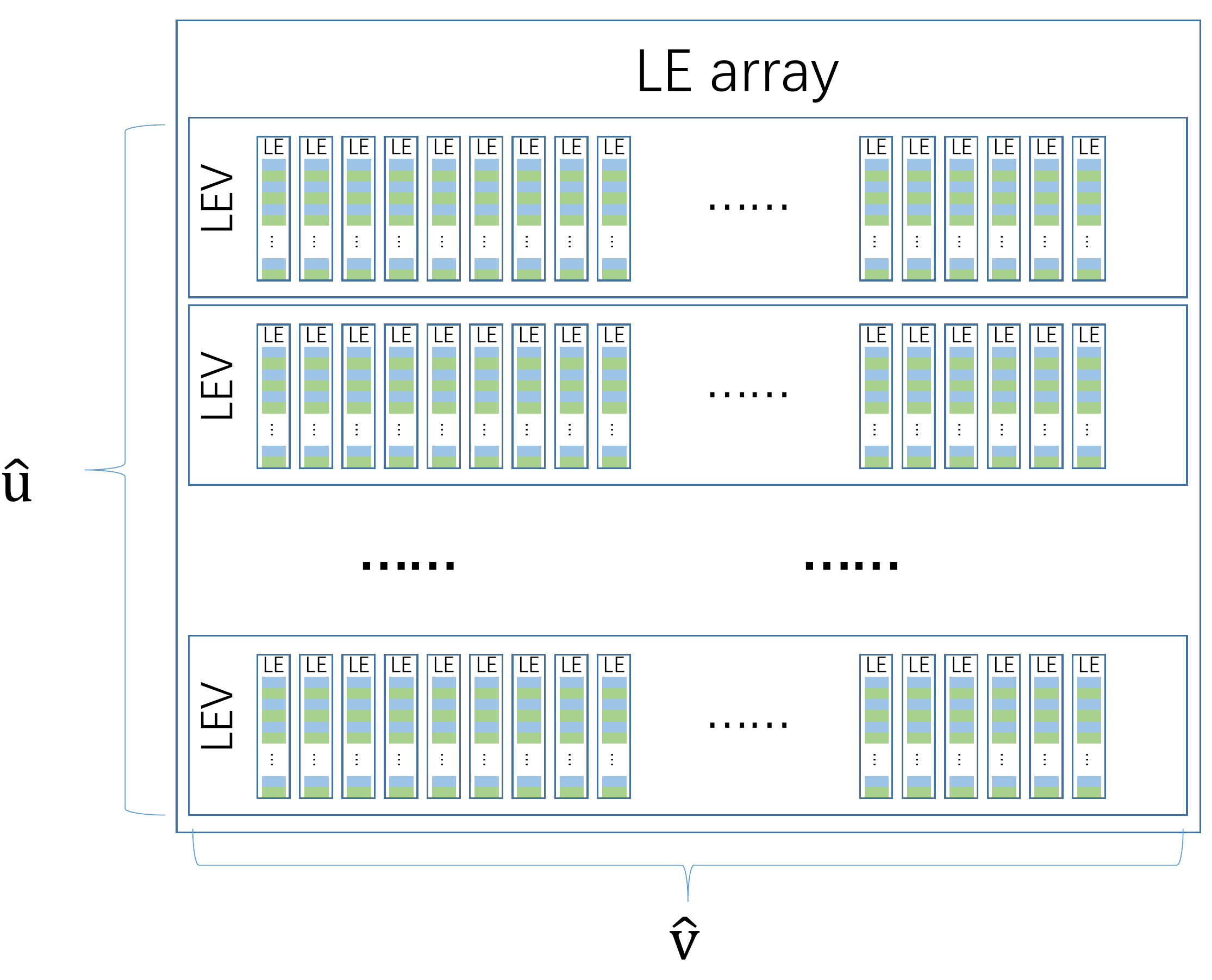}
\caption{Structure of LE array}
\label{fig_StructureOfLEArray}
\end{figure}

LE vector (LEV) contains $\LEArowNmath$ LE, and LEA contains $\RECrowNlgtmath$ LEV. Each observation node has a LEA, and the LEA at all observation nodes has the same structure. Let \LEAofOneNode{\SubnodeIdxmath} denote the LEA at the $\SubnodeIdxmath$-th observation node, and $\LEinLEAofOneNodemath{i}{j}{\SubnodeIdxmath}$ denote the $j$-th LE in the $i$-th LEV of \LEAofOneNode{\SubnodeIdxmath}.

For each $\aipmath$ in $\ANetmath$, READ selects one LE from each LEV of LEA to record the opposite hosts of $\aipmath$. READ maps $\aipmath$ to $\LEArowNmath$ LE in LEV with $\LEArowNmath$ random hash functions. 
READ uses the hash function \LEAhashFunction{i}{\aipmath} when mapping $\aipmath$ to a LE in the $i$-th LEV, where $\LEAhashFunctionmath{i}{\aipmath} \in [0,\LEAcolNmath-1], 0 \leq i \leq \LEArowNmath -1$. The observation node \SubnodeObj{\SubnodeIdxmath} not only updates \RECobj{\SubnodeIdxmath} but also \LEAofOneNode{\SubnodeIdxmath} when scanning \IPair{\TWinmath}{\SubnodeIdxmath}.

\begin{algorithm}         
\caption{scanIPair}  
\label{alg-READ_scanIPpair}  
\begin{algorithmic}[1]
\Require { $\RECaryNlgtmath$,$\RECrowNlgtmath$,
           $\{\RECcolNlgtmath{0}, \RECcolNlgtmath{0},\cdots, \RECcolNlgtmath{\RECrowNlgtmath-1} \}$, 
           $\{\REidxStartBitmath{0}, \REidxStartBitmath{1}, \cdots, \REidxStartBitmath{\RECrowNlgtmath-1}\}$, 
           $\LEArowNmath$ , $\LEAcolNmath$,
           $|\LEbitsSetmath|$,
           $\{\LEAhashFunctionNoParmath{0}, \LEAhashFunctionNoParmath{1},\cdots,\LEAhashFunctionNoParmath{\LEArowNmath-1}\}$,
           \IPair{\TWinmath}{\SubnodeIdxmath}} 
\Ensure{$\RECobjmath{\SubnodeIdxmath}$, \LEAofOneNode{\SubnodeIdxmath}} 
\State	Init \RECobj{\SubnodeIdxmath}  
\State	Init \LEAofOneNode{\SubnodeIdxmath}
\For {$<\aipmath,\bipmath> \in \IPairmath{\TWinmath}{\SubnodeIdxmath}$}
\State	$k \leftarrow$ right $\RECaryNlgtmath$ bits of $\aipmath$
\State	\IPLeftPart{\aipmath}$\leftarrow$ left 32-$\RECaryNlgtmath$ bits of $\aipmath$
	   \For {$i \in [0,$\RECrowNlgtmath$-1]$}
	     \State j=0
	       \For {$i_1 \in [0,\RECcolNlgtmath{i}-1$}
	         \State $j=j+ (\IPLeftPartmath{\aipmath}[(\REidxStartBitmath{i}+i_1 )mod(32-$\RECaryNlgtmath$)]<< i_1)$
	      \EndFor
	      \State Update $\REobjInRECmath{i,j,k}{l}$  with $\bipmath$ 
	    \EndFor
	    \For i∈[0,$\LEArowNmath$-1]
	        \State j= \LEAhashFunction{i}{\aipmath}
	        \State Update \LEinLEAofOneNode{i}{j}{\SubnodeIdxmath} with $\bipmath$ 
	    \EndFor
	 \EndFor
 \Return $\RECobjmath{\SubnodeIdxmath}$, \LEAofOneNode{\SubnodeIdxmath}

\end{algorithmic}
\end{algorithm}

The Algorithm \ref{alg-READ_scanIPpair} describes how READ scans IP address pairs in one observation node. READ first determines the size of REC and LEA according to the parameters, allocates the memory needed by REC and LEA, and initializes the counters of all RE and LE. Then start scanning each IP address pair in \IPair{\TWinmath}{\SubnodeIdxmath} and update REC and LEA. When scanning IP address pairs $<\aipmath, \bipmath>$, READ selects a REA from the REC by using $\RECaryNlgtmath$ bits on the right side of $\aipmath$, and  extracts $32-\RECaryNlgtmath$ bits on the left side of $\aipmath$ as \IPLeftPart{\aipmath}. Then the index of RE in each REV is determined according to \IPLeftPart{\aipmath}. Here, the index of RE refers to the location of RE in REV and takes the value between $[0,2^{\RECcolNlgtmath{i}}-1]$, where $2^{\RECcolNlgtmath{i}}$ is the number of RE contained in the REV. For the $i$-th REV, parameter \REidxStartBit{i} specifies the bits in \IPLeftPart{\aipmath} associating with the first bit of the RE index. 
After the index value of RE is obtained, the RE is updated with $\bipmath$. Compared with updating \RECobj{\SubnodeIdxmath}, updating \LEAofOneNode{\SubnodeIdxmath} is much simpler, because \LEAofOneNode{\SubnodeIdxmath} is only used to estimate the cardinality and does not need to restore the global candidate super point. 

After the observation node scans all IP address pairs in \IPair{\TWinmath}{\SubnodeIdxmath}, \RECobj{\SubnodeIdxmath} and \LEAofOneNode{\SubnodeIdxmath} record the information of opposite hosts. By collecting \RECobj{\SubnodeIdxmath} and \LEAofOneNode{\SubnodeIdxmath} from all observation nodes, the global candidate super points can be recovered and the cardinalities of candidate super points can be estimated.

The next section describes how READ recovers global candidate super points in a distributed environment.

\subsection{Generate candidate super points}
The master data structure at the observation node consists of two parts: REC and LEA. REC is used to recover global candidate super points, which has the advantage of less memory consumption; LEA is used to estimate cardinality, which has the advantage of high estimation accuracy. Each observation node can only observe part of the opposite hosts. In order to detect the super points accurately, it is necessary to collect the opposite hosts information recorded by each observation node on the global server. In this paper, we call the super points detected from IP address pairs of all observation nodes as global super points, and the generated global candidate super points as global candidate super points. When generating global candidate super points, only RECs are collected from each observation node, as shown in Figure \ref{fig_CollectRECFromObservatNodes}.

\begin{figure}[!ht]
\centering
\includegraphics[width=0.47\textwidth]{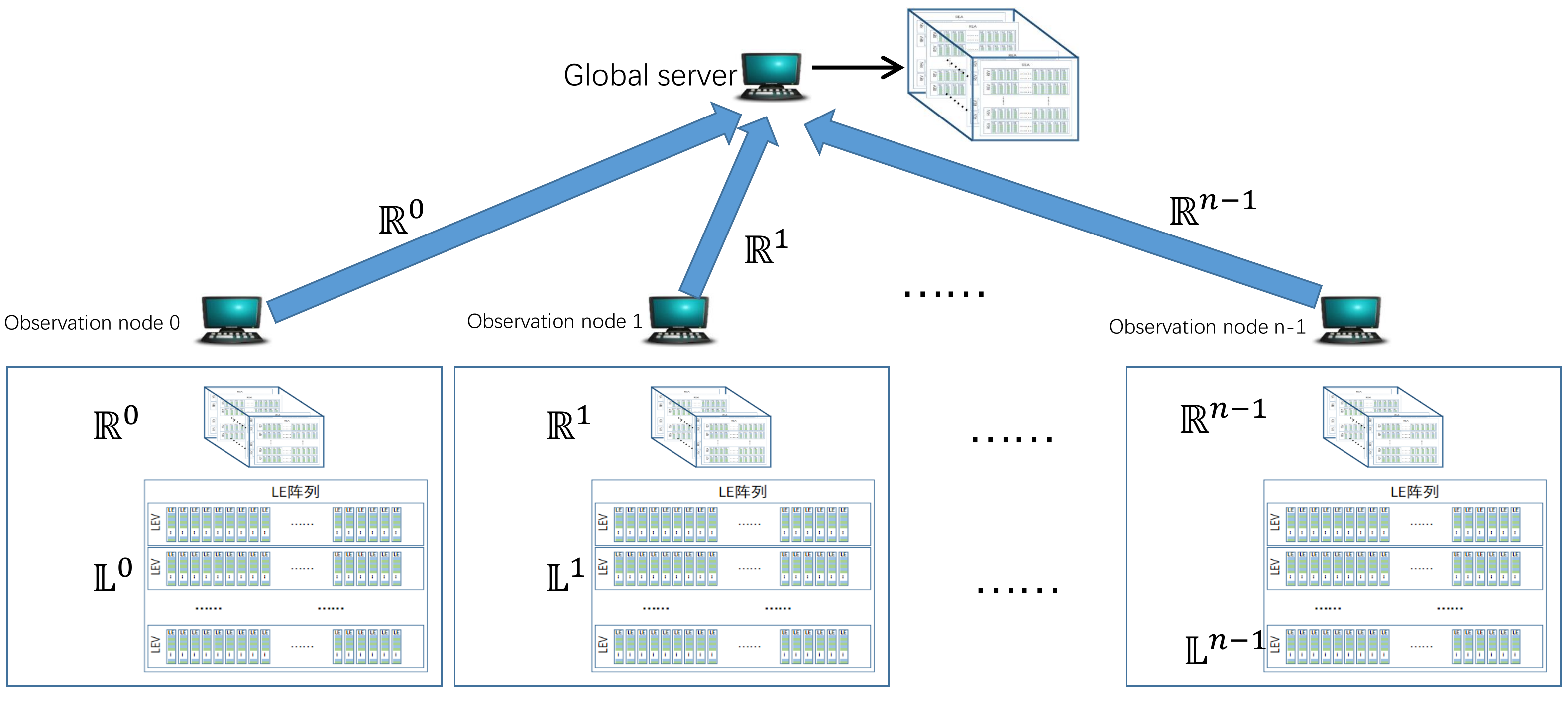}
\caption{Collect REC from observation nodes}
\label{fig_CollectRECFromObservatNodes}
\end{figure}

After each observation node has scanned all IP address pair in a time window, only the REC needs to be sent to the global server. The global server merges all the collected REC. The merging method is to merge the RE of different observation nodes in a ``bit or" manner. In this paper, the way of combining according to ``bit or" is called external merging, and the way of combining according to ``bit and" is called internal merging. External merger of RE is defined as follows:

\begin{definition}[RE Out merging]
\label{def-REOutMerging}
All bits of two RE generate a new RE according to the operation of ``bit or".
\end{definition}

In this paper, when the operand of the operator ``$\bigoplus$" is two RE or two LE, it means to out merge the two RE or LE; when the operand of the operator ``$\bigodot$" is two RE or two LE, it means to inner merge the two RE or LE.

The REC of all observation nodes are merged on the global server by outer merging, which ensures that any bit in the REC is still 1 in the merged global REC as long as it is set to 1 at any one observation node. Since RE uses bits to record the occurrence of opposite host, the global REC generated by outer merging contains the opposite information recorded by all observation nodes.

In this paper, we call the REC used to restore the global candidate super points on the global server as the global REC. The global REC has the same structure as the REC at all observation nodes. The global REC and the REC of all observation nodes are merged according to outer merging. There are two methods to get the global REC:
\begin{enumerate}[(1)]
\item Before merging the REC, the global server initializes a REC with the same structure as the REC at the observation nodes, and sets all bits in the initialized REC to 0. Then, the REC on the global server is merged with the REC on all observation nodes one by one, and the results are saved to the global REC. 
\item The global server takes the REC from the first observation node as the global REC, then merges the global REC with the REC from the remaining observation nodes, and saves the results to the global REC.
\end{enumerate}

Among the two methods for merging global REC, method (2) is less computational than method (1), because method (2) does not need to re-initialize REC. In this paper, method (2) is used to merge the REC of observation nodes into the global REC. Let $\RECobjonlymath$ denote the global REC, and $\RECobjonlymath_{i,j,k}$ denote the $j$-th RE of the $i$-th REV in the $k$-th REA of $\RECobjonlymath$. Assuming that the REC on \SubnodeObj{0} is first received one on the global server, Algorithm \ref{alg-OutMergingREC} describes the REC merging process on the global server.

\begin{algorithm}         
\caption{Out Merging REC}  
\label{alg-OutMergingREC}  
\begin{algorithmic}[1]
\Require { $\SubnodeNummath$, 
           $\{\RECobjmath{0}, \RECobjmath{1}, \cdots, \RECobjmath{n-1}\}$,
            $\RECaryNlgtmath$, $\RECrowNlgtmath$, $\{\RECcolNlgtmath{0}, \RECcolNlgtmath{1}, \cdots, \RECcolNlgtmath{\RECrowNlgtmath-1} \}$
            } 
\Ensure{$\RECobjonlymath$} 
	 \State $\RECobjonlymath \leftarrow \RECobjmath{0}$
	 \For {$\SubnodeIdxmath \in [1,\SubnodeNummath-1]$}
	    \For {$k \in [0,2^{\RECaryNlgtmath}-1]$}
	       \For {$ i \in [0,2^{\RECrowNlgtmath}-1]$}
	          \For {$j \in [0,2^{\RECcolNlgtmath{i}}-1]$}
	            \State $\RECobjonlymath_{i,j,k}\leftarrow (\RECobjonlymath_{i,j,k} \bigoplus \REobjInRECmath{i,j,k}{\SubnodeIdxmath} )$
	          \EndFor
	        \EndFor
	     \EndFor
	  \EndFor
\State Return $\RECobjonlymath$

\end{algorithmic}
\end{algorithm}

The first line of Algorithm \ref{alg-OutMergingREC} takes the received $\RECobjmath{0}$ as the global REC after the first merge, and then merges the remaining $\SubnodeNummath-1$ observation nodes into the global REC. After merging the REC at all observation nodes, algorithm \ref{alg-OutMergingREC} outputs the global REC.

READ recovers the global candidate super points from each REA of the global REC in turn. For the k-th REA of the global REC (denoted as \REaryObj{k}), READ calculates the global candidate super points in it by the following two steps:

\begin{enumerate}[Step (1)]
\item Find out all RE in \REaryObj{k} whose estimating cardinality is greater than the threshold.
\item From the candidate RE, 32-$\RECaryNlgtmath$ bits on the left of the candidate super point are recovered, and then concatenate with the right $\RECaryNlgtmath$ bits represented by k to get the complete global candidate super point.
\end{enumerate}
 
The above Step (1) only need to scan all RE in \REaryObj{k} once to get candidate RE. 
Let $\REcandTupleObjmath{i}=\{\REcandIdxObjmath{0}{i},\REcandIdxObjmath{1}{i},\REcandIdxObjmath{2}{i},\cdots\}$ represents the index of candidate RE in the $i$-th REV of \REaryObj{k}. Equation \ref{eq-READgetREidxFromLPofIP}  shows that the index of candidate RE in $\REcandTupleObjmath{i}$ comes from the bits of certain IP address. At the same time, as can be seen from Figure \ref{fig_getREidx_from_LP}, if the two indexes $\REcandIdxObjmath{x}{i}$ and $\REcandIdxObjmath{y}{((i+1)mod(\RECrowNlgtmath))}$ of two adjacent row, $i$ and $(i+1)mod(\RECrowNlgtmath)$ are from the same IP address, then  they have $\REidxStartBitmath{i}+\RECcolNlgtmath{i}-\REidxStartBitmath{((i+1)mod(\RECrowNlgtmath))}$ bits are the same. 
Conversely, if the left $\REidxStartBitmath{i}+\RECcolNlgtmath{i}-\REidxStartBitmath{((i+1)mod(\RECrowNlgtmath))}$ bits of $\REcandIdxObjmath{x}{i}$ are different from the right $\REidxStartBitmath{i}+\RECcolNlgtmath{i}-\REidxStartBitmath((i+1)mod(\RECrowNlgtmath))$ bits of $\REcandIdxObjmath{y}{((i+1)mod(\RECrowNlgtmath))}$, then $\REcandIdxObjmath{x}{i}$ and $\REcandIdxObjmath{y}{((i+1)mod(\RECrowNlgtmath))}$ certainly do not come from the same IP address. When the $\RECrowNlgtmath$ RE indexes comes from the same IP address, the $\RECrowNlgtmath$ RE indexes are called a candidate RE tuple. Inner merge these $\RECrowNlgtmath$ RE in a candidate RE tuple. If the estimated value of the inner merged RE still exceeds the threshold, the candidate RE tuple come from a global candidate hyper point.

When the candidate RE tuple comes from a global candidate super point, the candidate RE tuple can recover 32-$\RECaryNlgtmath$ bits to the left part of the global super point. From the setting requirement of parameter \REidxStartBit{i}, if the RE indexes in a candidate RE tuple comes from the same IP address $\aipmath$, any bit of \IPLeftPart{\aipmath} will appear at least once in the $\RECrowNlgtmath$ different candidate RE indexes. Therefore, 32-$\RECaryNlgtmath$ bits of \IPLeftPart{\aipmath} can be recovered from the candidate RE tuple. Then, a global candidate super point is obtained by concatenation with $k$, i.e. $(\IPLeftPartmath{\aipmath}<<\RECaryNlgtmath)+k$. 

Depth traversal can be used to calculate all candidate RE tuples from $\REcandTupleObjmath{i}$. 
For example, suppose that the parameters of REC are set to $\RECaryNlgtmath$ = 2, $\RECrowNlgtmath$ = 3, $\RECcolNlgtmath{0}=\RECcolNlgtmath{1}=\RECcolNlgtmath{2}=14$, \REidxStartBit{0}=0, \REidxStartBit{1}=10, \REidxStartBit{2}=20, 
the candidate RE indexes of $\REaryObjmath{2}$ 
is $\REcandTupleObjmath{0}=\{\REcandIdxObjmath{0}{0},\REcandIdxObjmath{1}{0},\REcandIdxObjmath{2}{0}\}, \REcandTupleObjmath{1}=\{\REcandIdxObjmath{0}{1},\REcandIdxObjmath{1}{1}\}$, $\REcandTupleObjmath{2}=\{\REcandIdxObjmath{0}{2},\REcandIdxObjmath{1}{2},\REcandIdxObjmath{2}{2}\}$. The number values of some candidate RE are as follows:

\begin{itemize}
\item $\REcandIdxObjmath{0}{0}=\greyBG{1100}010101\greyBG{0101}$
\item $\REcandIdxObjmath{0}{1}=\greyBG{1100}011001\greyBG{0101}$,\\
$\REcandIdxObjmath{1}{1}=\greyBG{1110}010001\greyBG{1100}$
\item $\REcandIdxObjmath{0}{2}=\greyBG{1001}011101\greyBG{1110}$,\\
 $\REcandIdxObjmath{1}{2}=\greyBG{0101}000101\greyBG{1110}$
\end{itemize}

In the above example, \REidxStartBit{i}+\RECcolNlgt{i}-\REidxStartBit{i+1}=4, that is, the candidate RE indexes in the two adjacent $\REcandTupleObjmath{i}$ determines whether it comes from the same IP address by the four bits on the left and the four bits on the right (the gray part in the RE index). When the candidate RE tuple is calculated by depth-first method, the candidate RE tuple is empty at the beginning, and then the first RE number is $\REcandIdxObjmath{0}{0}$. Test whether $\REcandIdxObjmath{0}{0}$ and $\REcandIdxObjmath{0}{1}$ come from the same IP address, as shown in Figure \ref{fig_exampleOfRestoringLP}.
\begin{figure}[!ht]
\centering
\includegraphics[width=0.47\textwidth]{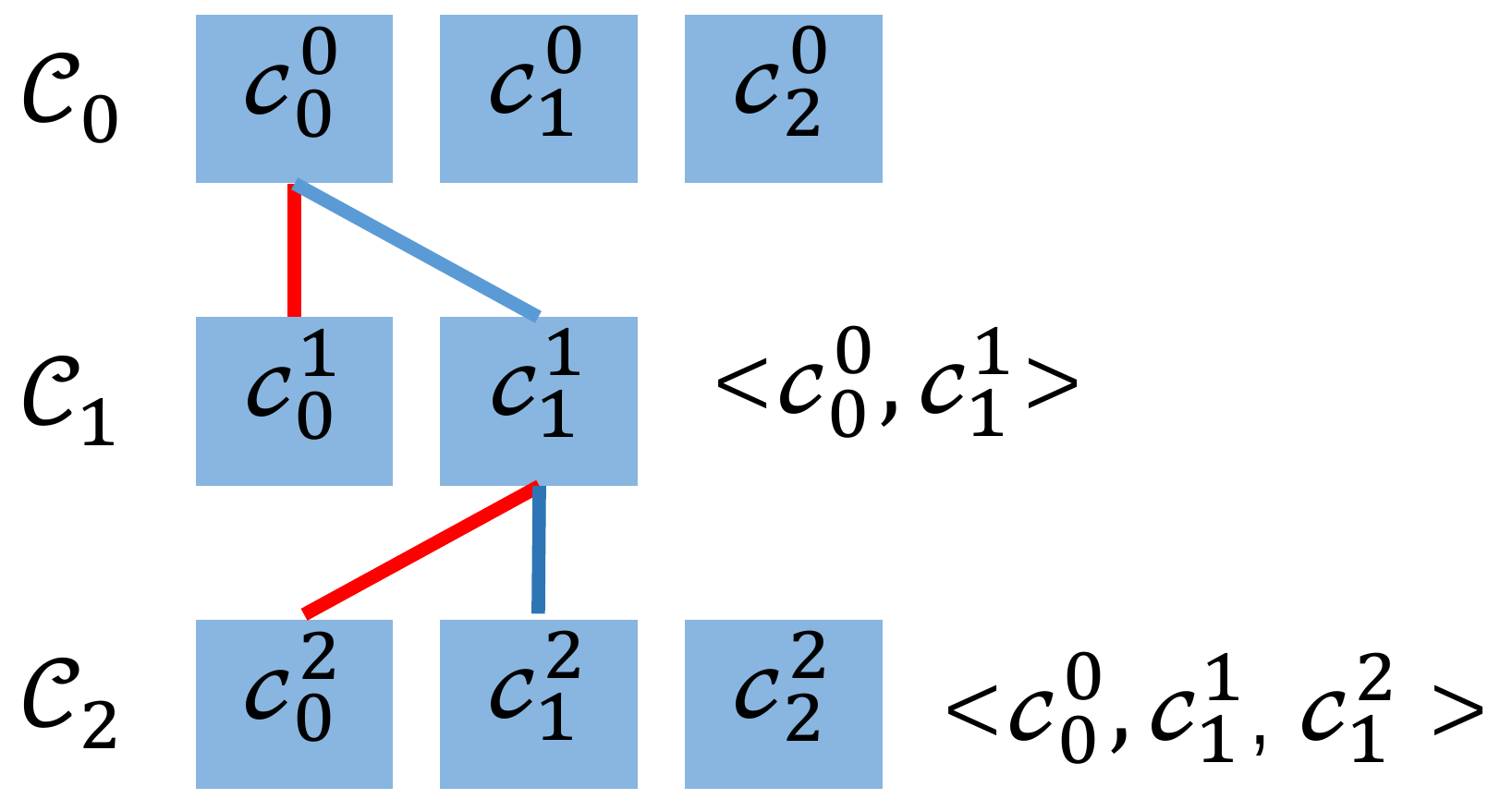}
\caption{Example of restoring LP with depth-first method}
\label{fig_exampleOfRestoringLP}
\end{figure}

The four bits on the left of $\REcandIdxObjmath{0}{0}$ are different from the four bits on the right of $\REcandIdxObjmath{0}{1}$, so $\REcandIdxObjmath{0}{0}$ and $\REcandIdxObjmath{0}{1}$ come from different IP addresses. Then test $\REcandIdxObjmath{0}{0}$ and $\REcandIdxObjmath{1}{1}$. The four bits on the left side of $\REcandIdxObjmath{0}{0}$ are the same as the four bits on the right side of $\REcandIdxObjmath{1}{1}$, so $\REcandIdxObjmath{1}{1}$ is added to the candidate RE tuple. Then find the RE index from $\REcandTupleObjmath{2}$ which come from the same IP address with $\REcandIdxObjmath{1}{1}$. In $\REcandTupleObjmath{2}$, the four bits on the right side of $\REcandIdxObjmath{0}{2}$ are the same as the four bits on the left side of $\REcandIdxObjmath{1}{1}$, but the four bits on the left side of $\REcandIdxObjmath{0}{2}$ are not equal to the four bits on the right side of $\REcandIdxObjmath{0}{0}$, so $\REcandIdxObjmath{0}{2}$ cannot form a candidate RE tuple with $\REcandIdxObjmath{0}{0}$ and $\REcandIdxObjmath{1}{1}$. In $\REcandTupleObjmath{2}$, not only are the four bits on the right side the same as the four bits on the left side of $\REcandIdxObjmath{1}{1}$, but also the four bits on the left side of $\REcandIdxObjmath{1}{2}$ the same as the four bits on the right side of $\REcandIdxObjmath{0}{0}$. Therefore, $<\REcandIdxObjmath{0}{0}$,$\REcandIdxObjmath{1}{1}$,$\REcandIdxObjmath{1}{2}>$ constitutes a candidate RE tuple. 

From the values of $\REcandIdxObjmath{0}{0}$, $\REcandIdxObjmath{1}{1}$ and $\REcandIdxObjmath{1}{2}$, we can see that the RE associating with the candidate RE tuple is $\REobjInRECmath{0,12629,2}{\SubnodeIdxmath}$, $\REobjInRECmath{1,14620,2}{\SubnodeIdxmath}$ , $\REobjInRECmath{2,5214,2}{\SubnodeIdxmath}$. If the cardinality estimated from the inner merge RE, $\REobjInRECmath{0,12629,2}{\SubnodeIdxmath} \bigodot \REobjInRECmath{1,14620,2}{\SubnodeIdxmath} \bigodot \REobjInRECmath{2,5214,2}{\SubnodeIdxmath}$, still over the threshold, 30 bits of the left part of $\aipmath$ can be recovered from $<\REcandIdxObjmath{0}{0},\REcandIdxObjmath{1}{1},\REcandIdxObjmath{1}{2}>$: ``000101\greyBG{1110}010001\greyBG{1100}010101\greyBG{0101}". 
$\REaryObjmath{2}$ is the 2-th REA in REC. The associating binary format is ``10". Thus, the global candidate super point is ``000101\greyBG{1110}010001\greyBG{1100}010101\greyBG{0101}10". 

All REA in global REC are processed in the above way. Because the number of RE counters is small (for IPv4 address, there are only 8 counters), so it is faster to scan REA and calculate candidate RE number. And each RE only takes up one byte of space, so REC takes up less memory and reduces the amount of data transmitted between observation nodes and the global server. However, the cardinalities of the global candidate super points cannot be estimated by RE. Estimating the cardinality requires the use of the opposite host information stored in LEA. The next section describes how to collect the opposite host information stored in LEA from the observation nodes, estimate the cardinalities of the global candidate super points, and filter out the super points.

\subsection{Estimate cardinalities of candidate super points}
The LEA at each observation node is used for estimating the cardinality of global candidate super points. A simple way is to send all LEAs at each observation node to the global server, and then merge all LEA of observation nodes on the global server in a ``bit or" manner to get the global LEA. 

In this paper, when the operand of ``$\sum$" is the LE or RE set, it means that all LE or RE in the set are merged by outer merging method; when the operand of ``$\prod$" is the LE or RE set, it means that all LE or RE in the set are merged by inner merging method.

Merging LEA of all observation nodes on the global server in the way of outer merging is equivalent to sending IP address pairs directly to the global server to update the global LEA. Because LE outer merging guarantees that any bit in the global LEA will remain 1 as long as it is set to 1 at one or more observation nodes.

After the global LEA is generated, the cardinalities of global candidate super points can be estimated according to the global LEA. Let $\CadIPobjmath$ denote a global candidate super point, $\LEcipInSLEAmath{i}{\SubnodeIdxmath}{\CadIPobjmath}$ denote the LE of $\CadIPobjmath$ in the $i$-th LEV of the $\SubnodeIdxmath$-th observation node, i.e. $\LEcipInSLEAmath{i}{\SubnodeIdxmath}{\CadIPobjmath}=\LEinLEAofOneNodemath{i}{j}{\SubnodeIdxmath}$, $j=\LEAhashFunctionmath{i}{\CadIPobjmath}$. Using hash functions $\LEAhashFunctionmath{i}{\CadIPobjmath}$, we can find these LE used by $\CadIPobjmath$ from the global LEA. 

Let $\LEcipInGLEAmath{i}{\CadIPobjmath}$ denote the LE associating with $\CadIPobjmath$ in the first LEV of the global LEA. Since global LEA is obtained by combining LEA from all observation nodes, $\LEcipInGLEAmath{i}{\CadIPobjmath}=\sum_{l=0}^{n-1}\LEcipInSLEAmath{i}{\SubnodeIdxmath}{\CadIPobjmath}$ .
The $\LEArowNmath$ LE of $\CadIPobjmath$ on the global LEA are merged into $\ULEcipInGLEAmath{\CadIPobjmath}=\prod_{i=0}^{\LEArowNmath-1}{\LEcipInGLEAmath{i}{\CadIPobjmath}}$ . 
Let $|\ULEcipInGLEAmath{\CadIPobjmath}|$ denote the number of bits with value ``1" in $\ULEcipInGLEAmath{\CadIPobjmath}$. The cardinality of $\CadIPobjmath$ is estimated based on $\ULEcipInGLEAmath{\CadIPobjmath}$ by equation \ref{eq-LE-cardinalityEstimation}. If the estimated result is larger than the threshold, $\CadIPobjmath$ is reported as a super point. 

Although the above method avoids sending all IP addresses to the global server, it still needs to send the complete LEA to the global server. In order to improve the accuracy of cardinality estimating, the parameters of LEA are set to larger values. For example, when $\LEArowNmath=5$, $\LEAcolNmath=2^{15}$, $|\LEbitsSetmath|=2^{14}$, LEA is 320 MB in size. When estimate cardinalities, each observation node needs to send 320MB of data to the global server.

When estimating the cardinality of global candidate super point $\CadIPobjmath$, only $\ULEcipInGLEAmath{\CadIPobjmath}$ is needed.
Based on this principle, READ first sends the global candidate super points to each observation node from the global server, and then each observation node send these LE relating with candidate super points back to the global server, as shown in Figure \ref{fig_collectCandidateIPfromLE}.

\begin{figure}[!ht]
\centering
\includegraphics[width=0.47\textwidth]{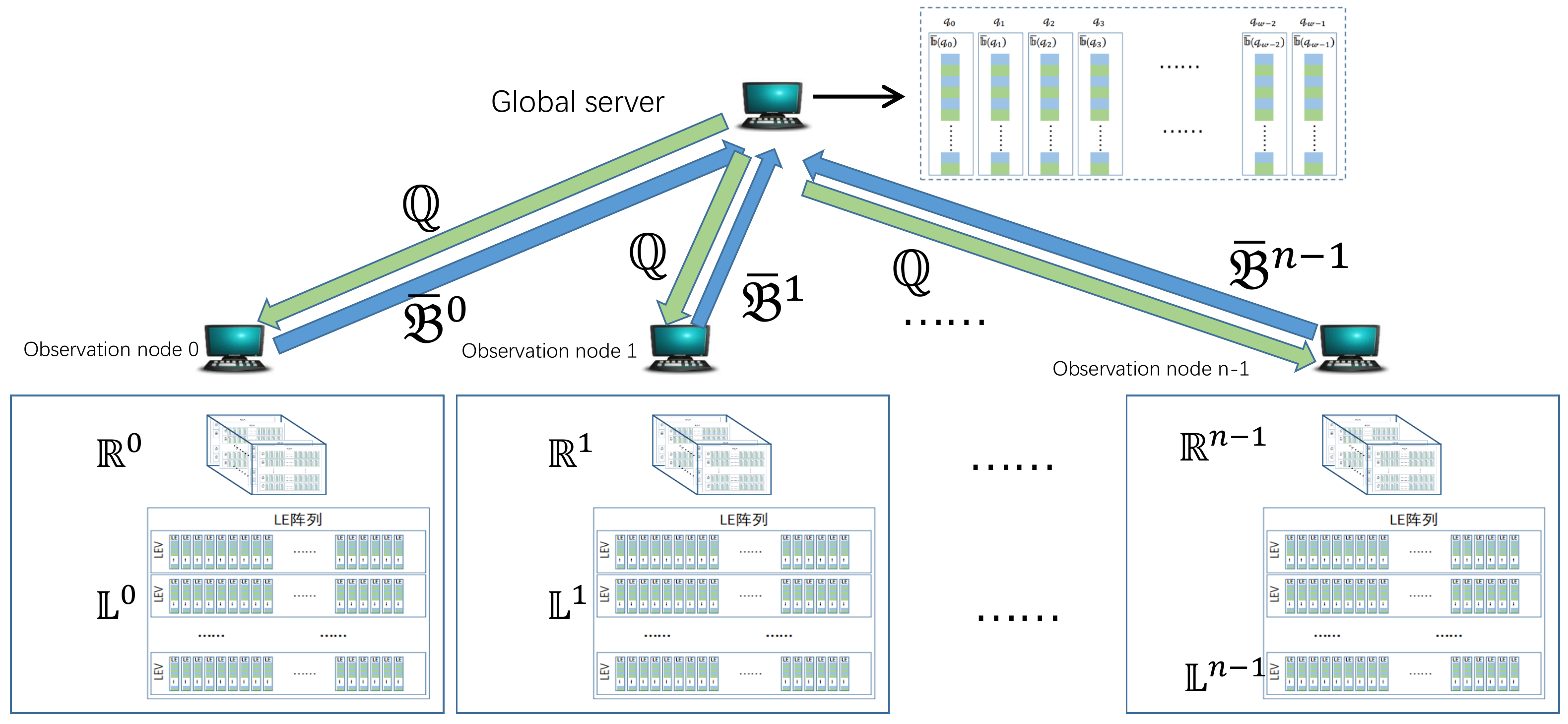}
\caption{ Super point detection in distributed environmentr}
\label{fig_collectCandidateIPfromLE}
\end{figure}

In Figure \ref{fig_collectCandidateIPfromLE}, $\CadIPobjSetmath=\{\CadIPobjmath_{0}, \CadIPobjmath_{1}, \CadIPobjmath_{2}, \cdots ,\CadIPobjmath_{\CIPnumbermath-1}\}$ denotes the set of global candidate super points, $\ULEcipsetInSNodemath{\SubnodeIdxmath}$ denotes the set of LE used to estimate cardinalities of global candidate super points in $\CadIPobjSetmath$ on the observation node $\SubnodeIdxmath$. 
For global candidate super point $\CadIPobjmath$, there are $\LEArowNmath$ LE associating with it, i.e. $\{\LEcipInSLEAmath{0}{\SubnodeIdxmath}{\CadIPobjmath},\LEcipInSLEAmath{1}{\SubnodeIdxmath}{\CadIPobjmath},\cdots,\LEcipInSLEAmath{\LEArowNmath-1}{\SubnodeIdxmath}{\CadIPobjmath}\}$.
READ does not send all of the $\LEArowNmath$ LE to the global server, but the result of internal merging , $\ULEcipInSNodemath{\SubnodeIdxmath}{\CadIPobjmath} =\prod_{i=0}^{\LEArowNmath-1}\LEcipInSLEAmath{i}{\SubnodeIdxmath}{\CadIPobjmath}$.
In Figure \ref{fig_collectCandidateIPfromLE}, $\ULESetofcipInSNodemath{\SubnodeIdxmath}=\{\ULEcipInSNodemath{\SubnodeIdxmath}{\CadIPobjmath_{0}},\ULEcipInSNodemath{\SubnodeIdxmath}{\CadIPobjmath_{1}},\ULEcipInSNodemath{\SubnodeIdxmath}{\CadIPobjmath_{2}},\cdots,\ULEcipInSNodemath{\SubnodeIdxmath}{\CadIPobjmath_{\CIPnumbermath -1}} \}$ is the LE set to be sent to the global server on the $\SubnodeIdxmath$-th observation node. 

On the global server, $\ULEcipInGLEAREADmath{\CadIPobjmath}=\sum_{\SubnodeIdxmath=0}^{\SubnodeNummath}\ULEcipInSNodemath{\SubnodeIdxmath}{\CadIPobjmath}$, which is used for estimating the cardinality of $\CadIPobjmath$, is obtained by outer merging all $\ULEcipInSNodemath{\SubnodeIdxmath}{\CadIPobjmath}$. Let $|\ULEcipInGLEAREADmath{\CadIPobjmath}|$ denote the number of bits with value ``1" in $\ULEcipInGLEAREADmath{\CadIPobjmath}$. Theorem \ref{thm-DistributedLE_close_toLE_cnt} shows that $\ULEcipInGLEAREADmath{\CadIPobjmath}$ can more accurately estimate the cardinality of $\CadIPobjmath$ than $\ULEcipInGLEAmath{\CadIPobjmath}$.

\begin{theorem}
\label{thm-DistributedLE_close_toLE_cnt}
For global candidate super point $\CadIPobjmath$, let $\OPmath{\TWinmath}{\CadIPobjmath}$ denote the set of opposite hosts of $\CadIPobjmath$ passing through all observation nodes in time window $\TWinmath$, $\LEusedClusivemath{\CadIPobjmath}$ denote a LE after scanning $\OPmath{\TWinmath}{\CadIPobjmath}$, and $|\LEusedClusivemath{\CadIPobjmath}|$ denote the number of bits with value ``1" in $\LEusedClusivemath{\CadIPobjmath}$. Then these bits with value ``1" in $\LEusedClusivemath{\CadIPobjmath}$ are still with value ``1" in $\ULEcipInGLEAREADmath{\CadIPobjmath}$ and $\ULEcipInGLEAmath{\CadIPobjmath}$. And $|\LEusedClusivemath{\CadIPobjmath}|\leq |\ULEcipInGLEAREADmath{\CadIPobjmath}|\leq |\ULEcipInGLEAmath{\CadIPobjmath}|$.
\end{theorem}
\begin{proof}
When a bit in $\LEusedClusivemath{\CadIPobjmath}$ has value ``1", there exists an IP address pair $<\CadIPobjmath,\bipmath>$ in $\OPmath{\TWinmath}{\CadIPobjmath}$ to set the bit to ``1". In global LEA, $\bipmath$ sets all the bits of $\LEArowNmath$ LE associating with $\CadIPobjmath$. After inner merging in LE, the bit is ``1" in $\ULEcipInGLEAmath{\CadIPobjmath}$. At the same time, $\bipmath$ will appear on at least one observation node and set all the bits of $\LEArowNmath$ LE associating with $\CadIPobjmath$ to ``1". 
Since the bit is ``1" in at least one $\ULEcipInSNodemath{\SubnodeIdxmath}{\CadIPobjmath}$, the bit is still ``1" after outer merging on the global server. So $|\LEusedClusivemath{\CadIPobjmath}| \leq |\ULEcipInGLEAmath{\CadIPobjmath}|$ and $|\LEusedClusivemath{\CadIPobjmath}| \leq |\ULEcipInGLEAREADmath{\CadIPobjmath}|$. 
Then we will proof that $|\ULEcipInGLEAREADmath{\CadIPobjmath}|\leq |\ULEcipInGLEAmath{\CadIPobjmath}|$.

Let $\LEcipInGLEAmath{i}{\CadIPobjmath}=\sum_{\SubnodeIdxmath=0}^{n-1}\LEcipInSLEAmath{i}{\SubnodeIdxmath}{\CadIPobjmath}$ , then $\ULEcipInGLEAmath{\CadIPobjmath}=\prod_{i=0}^{\LEArowNmath-1}\LEcipInGLEAmath{i}{\CadIPobjmath} =\prod_{i=0}^{\LEArowNmath-1}\sum_{l=0}^{n-1}\LEcipInSLEAmath{i}{\SubnodeIdxmath}{\CadIPobjmath} $.
Let $\ULEcipInSNodemath{\SubnodeIdxmath}{\CadIPobjmath}=\prod_{i=0}^{\LEArowNmath-1}\LEcipInSLEAmath{i}{\SubnodeIdxmath}{\CadIPobjmath}$, then $\ULEcipInGLEAREADmath{\CadIPobjmath}=\sum_{\SubnodeIdxmath=0}^{n-1}\ULEcipInSNodemath{\SubnodeIdxmath}{\CadIPobjmath} =\sum_{\SubnodeIdxmath=0}^{n-1}\prod_{i=0}^{\LEArowNmath-1}\LEcipInSLEAmath{i}{\SubnodeIdxmath}{\CadIPobjmath}$. 
To proof that $|\ULEcipInGLEAREADmath{\CadIPobjmath}|\leq |\ULEcipInGLEAmath{\CadIPobjmath}|$ is equivalent to proof that the number of bits with value ``1" in $\sum_{\SubnodeIdxmath=0}^{n-1}\prod_{i=0}^{\LEArowNmath-1}\LEcipInSLEAmath{i}{\SubnodeIdxmath}{\CadIPobjmath}$  is no more than the number of bits with value ``1" in $\prod_{i=0}^{\LEArowNmath-1}\sum_{\SubnodeIdxmath=0}^{n-1}\LEcipInSLEAmath{i}{\SubnodeIdxmath}{\CadIPobjmath}$. 
$\LEcipInSLEAmath{i}{\SubnodeIdxmath}{\CadIPobjmath}$ is a LE and the number of bits in all $\LEcipInSLEAmath{i}{\SubnodeIdxmath}{\CadIPobjmath}$ are the same. Let $\bitInArymath{i}{\SubnodeIdxmath}$ denote an arbitrary bit in $\LEcipInSLEAmath{i}{\SubnodeIdxmath}{\CadIPobjmath}$.
All $\bitInArymath{i}{\SubnodeIdxmath}$ in different observation nodes could be written as an array in the following format: 
$$
\bitArraySetmath = \begin{bmatrix}
\bitInArymath{0}{0} & \cdots & \bitInArymath{0}{\SubnodeNummath-1}\\ 
 \vdots & \ddots  & \vdots\\ 
\bitInArymath{\LEArowNmath-1}{0} & \cdots & \bitInArymath{\LEArowNmath-1}{\SubnodeNummath-1}
\end{bmatrix}
$$ 

In $\bitArraySetmath$, $\prod_{i=0}^{\LEArowNmath-1}\sum_{\SubnodeIdxmath=0}^{n-1}\bitInArymath{i}{\SubnodeIdxmath}$  represents that ``bit or" operations are performed on each line, and then ``bit and" operations are performed on the results; $\sum_{\SubnodeIdxmath=0}^{n-1}\prod_{i=0}^{\LEArowNmath-1}\bitInArymath{i}{\SubnodeIdxmath}$  represents that ``bit and" operations are performed on each line, and then ``bit or" operations are performed on the results. 

When $\prod_{i=0}^{\LEArowNmath-1}\sum_{\SubnodeIdxmath=0}^{n-1}\bitInArymath{i}{\SubnodeIdxmath}=0$, at least one row has all bits equal to ``0", and the result of ``bit and" operation for each column is also 0, then $\sum_{\SubnodeIdxmath=0}^{n-1}\prod_{i=0}^{\LEArowNmath-1}\bitInArymath{i}{\SubnodeIdxmath}=\sum_{\SubnodeIdxmath=0}^{n-1}0=0$. 
When $\prod_{i=0}^{\LEArowNmath-1}\sum_{\SubnodeIdxmath=0}^{n-1}\bitInArymath{i}{\SubnodeIdxmath}=1$, there is no row whose bits are all ``0". 
But $\sum_{\SubnodeIdxmath=0}^{n-1}\prod_{i=0}^{\LEArowNmath-1}\bitInArymath{i}{\SubnodeIdxmath}$ may still be 0. 
Because when each column of $\bitArraySetmath$ contains at least one bit with value ``0", then $\sum_{\SubnodeIdxmath=0}^{n-1}\prod_{i=0}^{\LEArowNmath-1}\bitInArymath{i}{\SubnodeIdxmath}=\sum_{\SubnodeIdxmath=0}^{n-1}0=0$. At this time, each row may contains one or more bits with value ``1". 
For example, when n=3,$\LEArowNmath =3$,$\bitArraySetmath = \begin{bmatrix}
1 & 0 & 0 \\ 
0 & 1 & 0 \\ 
0 & 0 & 1
\end{bmatrix}$, $\prod_{i=0}^{\LEArowNmath-1}\sum_{\SubnodeIdxmath=0}^{n-1}\bitInArymath{i}{\SubnodeIdxmath}=1 $, but$\sum_{\SubnodeIdxmath=0}^{n-1}\prod_{i=0}^{\LEArowNmath-1}\bitInArymath{i}{\SubnodeIdxmath}=0$.

When $\sum_{\SubnodeIdxmath=0}^{n-1}\prod_{i=0}^{\LEArowNmath-1}\bitInArymath{i}{\SubnodeIdxmath}=1$, $\prod_{i=0}^{\LEArowNmath-1}\sum_{\SubnodeIdxmath=0}^{n-1}\bitInArymath{i}{\SubnodeIdxmath}$  also equals to 1. 
Because when $\sum_{\SubnodeIdxmath=0}^{n-1}\prod_{i=0}^{\LEArowNmath-1}\bitInArymath{i}{\SubnodeIdxmath}=1$, at least one column in $\bitArraySetmath$ has all bits with value ``1". Then there is no row in $\bitArraySetmath$ whose bits are all ``0". Becuae $\bitInArymath{i}{\SubnodeIdxmath}$ is an arbitrary bit in $\LEcipInSLEAmath{i}{\SubnodeIdxmath}{\CadIPobjmath}$, then:
\begin{itemize}
\item When a bit has value ``1" in $\ULEcipInGLEAREADmath{\CadIPobjmath}$, the bit has value ``1" in $\ULEcipInGLEAmath{\CadIPobjmath}$; 
\item When a bit has value ``0" in $\ULEcipInGLEAmath{\CadIPobjmath}$, the bit has value ``0" in $\ULEcipInGLEAREADmath{\CadIPobjmath}$;
\item When a bit has value ``1" in $\ULEcipInGLEAmath{\CadIPobjmath}$, the bit may has value ``0" in $\ULEcipInGLEAREADmath{\CadIPobjmath}$
\end{itemize}

So the number of bits with value ``1" in $\ULEcipInGLEAREADmath{\CadIPobjmath}$ is no more than that in $\ULEcipInGLEAmath{\CadIPobjmath}$ and $|\LEusedClusivemath{\CadIPobjmath}|\leq |\ULEcipInGLEAREADmath{\CadIPobjmath}|\leq |\ULEcipInGLEAmath{\CadIPobjmath}|$.
\end{proof}

LE estimates cardinality based on the number of bits with value ``1". Theorem \ref{thm-DistributedLE_close_toLE_cnt} shows that the number of bits with value ``1" in $\ULEcipInGLEAREADmath{\CadIPobjmath}$ is closer to the number of bits with value ``1" in the LE which is used by $\CadIPobjmath$ exclusively. So the accuracy of estimating cardinality by $\ULEcipInGLEAREADmath{\CadIPobjmath}$ is more accuracy.

READ not only does not need to transfer the entire LEA to the global server, but also has a higher accuracy in estimating cardinalities of global candidate super points. When estimating cardinalities, the amount of data transmitted between each observation node and the global server is $(32*\CIPnumbermath+|\LEbitsSetmath|*\CIPnumbermath)$ bits, where $\CIPnumbermath$ is the number of candidate super points recovered by REC. $32*\CIPnumbermath$ is the data size of global candidate super points transmitting to each observation node from the global server, and $|\LEbitsSetmath|*\CIPnumbermath$ is the data size of LE of candidate super points that transmitting to the global server from each observation node. When $(32*\CIPnumbermath+|\LEbitsSetmath|*\CIPnumbermath)<\LEArowNmath*\LEAcolNmath*|\LEbitsSetmath|$, the data transmission between an observation node and the global server is less than the data transmission of the entire LEA. Global candidate super points account for only a small portion of all IP addresses, usually hundreds to thousands. In order to improve the estimation accuracy, the value of $\LEArowNmath *\LEAcolNmath$ will be more than tens of thousands. So READ reduces the amount of data transmitted between observation nodes and the global server. READ can also apply more powerful counters to replace bits in RE and LE to realize the detection of super points under sliding time window as discussed in the next section.

\section{Distributed super points detection under sliding time window} \label{sec-READunderSlidingWindow}
READ only scans IP address pairs at each observation node, so only sliding window counter is needed to record opposite hosts incrementally at the observation node. The master data structure at the observation node consists of two parts: REC and LEA. The estimators of REC and LEA are RE and LE, while the counters used by RE and LE are bits. So the master data structure at the observation node can be regarded as a set of bits. Using counter DR\cite{ispa2017:highspeednetworksuperpointsdetectionbasedslidingwindowgpu} or AT\cite{VATE2018:sw} under sliding window instead of bit in REC and LEA at each observation node, distributed super point detection under sliding window can be realized.

The counter under the sliding window needs to be updated. After all LE associating with the global candidate super points are sent to the global server, the observation node can start to update the sliding counter. At the end of each time window, the REC on the global server is generated by these REC collecting from all observation nodes, there is no need to update it.

Under the sliding time window, the observation node only needs to send the active state of the counter to the global server, that is, at the end of the time window, each sliding window counter can be changed into a bit: 0 for inactivity, 1 for activity. Therefore, under sliding time window, the traffic between observation nodes and the global server is the same as that under discrete time window.

\section{Experiments and analysis} \label{sec-experiments}
In order to test the performance of READ, four groups of high-speed network traffic are used to carry out experiments in this section. The experiment analyzes READ from the aspects of detection error rate, memory usage and running time. We compared READ with DCDS, VBFA, CSE and SRLA.
\subsection{Experiment data}
In this paper, four groups of high-speed network traffic are used. Two of the four sets of data come from the 10 Gb/s Caida\cite{expdata:caida}. The other two groups are from the network boundary of the 40Gb/s CERNET in Nanjing network\cite{expdata:iptracecernetjs:en}. 

The Caida data acquisition dates are February 19, 2015 and January 21, 2016 (denoted by $Caida\ 2015\_2\_19$ and $Caida\ 2016\_01\_21$), and the data acquisition dates of the two groups of CERNET Nanjing network were October 23, 2017 and March 8, 2018 (denoted by $IPtas\ 2017\_10\_23$ and $IPtas\ 2018\_03\_08$). The collection time of the four groups of data is one hour from 13:00. The collected data are raw IP Trace. Caida data collected Trace between Seattle and Chicago. In this paper, the IP on Seattle side is defined as $\aipmath$, and the IP on Chicago side is defined as $\bipmath$. IPtas data collects traces between CERNET Nanjing network and other networks. In this paper, the IP in Nanjing network is $\aipmath$, and in the other network is $\bipmath$.

In the experiment of this section, the length of time window is 5 minutes, and the threshold of super point is set to 1024. Therefore, each group of experimental data contains 12 time windows. Table \ref{tbl_exp_data_static} lists the statistical information of each experimental data. The number of $\aipmath$ in Caida data is more than the number of $\aipmath$ in IPtas data, which is 1.85 times more on average. But the average cardinality per $\aipmath$ in Caida data is less than that in IPtas data, only $21.389\%$ of the latter. The number of packets per second determines the number of IP address pairs that need to be processed per second. Therefore, packet speed (in millions of packets per second, Mpps) is a key attribute. As can be seen from Table 1, the average packet speed of IPtas data is 3.89 times that of Caida data. Therefore, Caida data and IPtas data represent two different types of network data sets, which can test the effect of the algorithm more comprehensively.

\begin{table}
\centering
\caption{Statistics of experiment data}
\label{tbl_exp_data_static}
\begin{tabular}{c}                                                                                                                                                                                                                           
\centering
\includegraphics[width=0.47\textwidth]{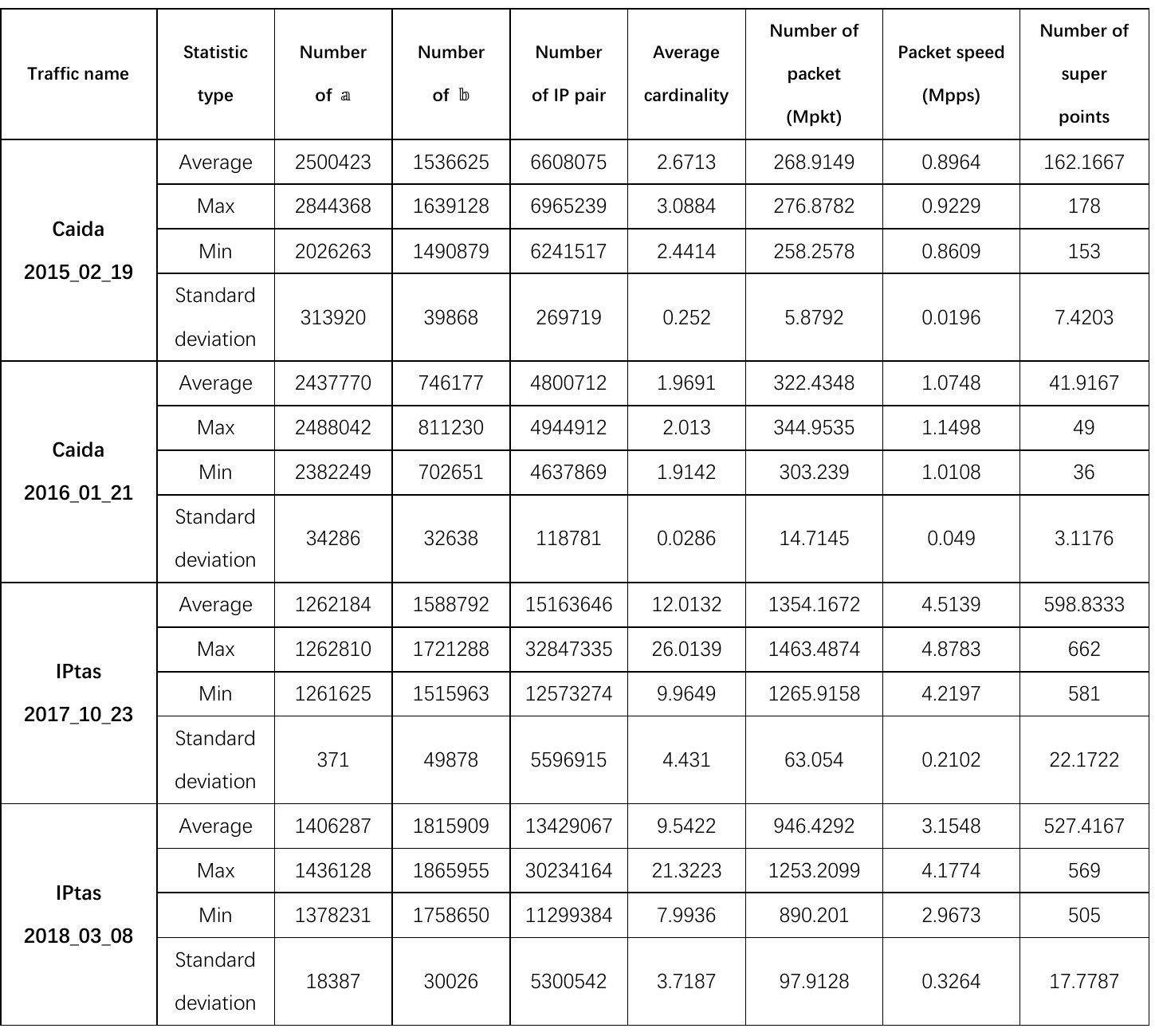}
\end{tabular}
\end{table}

\subsection{The purpose and scheme of the experiment}
The experimental purposes of this paper are as follows:
\begin{itemize}
\item Analyse the accuracy of READ and test whether REC can accurately generate candidate super points.
\item Analyzing the memory occupancy and running time of READ;
\item Test the number of candidate super points generated by READ and the amount of data that needs to be transmitted between each observation node and the global server.
\end{itemize}

In order to process high-speed network data in real time, this paper deploys READ, DCDS, VBFA, CSE and SRLA algorithm on GPU platform. All the experiments in this paper run on a server with GPU. The running environment is: Intel Xeon E5-2643 CPU, 125 GB memory, Nvidia Titan XP GPU, 12 GB memory, Debian Linux 9.6 operating system.

In the experiment, the parameters of REC are $\RECaryNlgtmath = 6$, $\RECrowNlgtmath$ = 3, $\RECcolNlgtmath{0}=\RECcolNlgtmath{1}=\RECcolNlgtmath{2}=14$; the parameters of LEA are $\LEArowNmath=5$,$\LEAcolNmath=2^{15}$ and $|\LEbitsSetmath|=2^{15}$. From the above parameters, it can be seen that REC occupies 3 MB of memory and LEA occupies 320 MB of memory. Because there is no distributed experimental data, the experiment in this section is carried out under a single node. However, from the previous analysis of READ, we can see that the error rate of READ in distributed environment will not be higher than that in single node environment.

\subsection{Memory and false rate}
In order to analyze the memory and false rate of READ, this section compares READ with DCDS, VBFA, CSE and SRLA algorithm. Table \ref{tbl_exp_memoryAndFalseRate} shows the average memory occupancy and error rate of READ and comparison algorithms in different experimental data sets. False positive rate (FPR), false negative rate (FNR) and false total rate (FTR) are three kind of false rate. Let N represent the number of super points, $N^-$ represent the number of super points that not detected out by an algorithm and $N^+$ represent the number of hosts whose cardinalities are less than the threshold but detected as super points by an algorithm. Then $FPR=100*N^+/N\%$, $FNR=100*N^-/N\%, FTR=FPN+FNR$.

\begin{table}
\centering
\caption{Memory and false rate}
\label{tbl_exp_memoryAndFalseRate}
\begin{tabular}{c}                                                                                                                                                                                                                           
\centering
\includegraphics[width=0.47\textwidth]{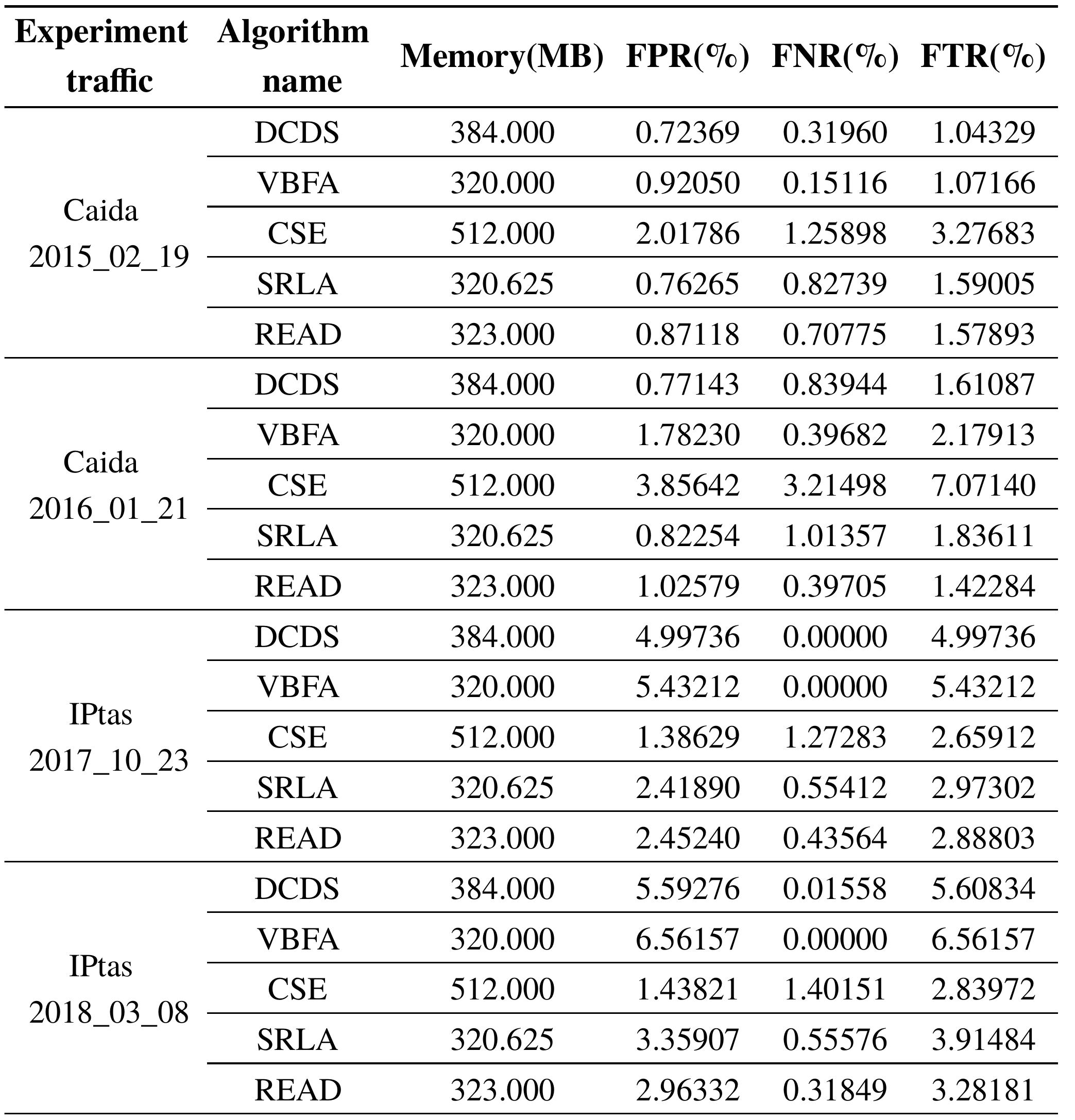}
\end{tabular}
\end{table}

Table \ref{tbl_exp_memoryAndFalseRate} shows that READ occupies less memory than DCDS and CSE, and only 3 MB more memory than VBFA. In terms of error rate, the error rate of READ is close to that of SRLA algorithm.

\subsection{Running time analysis}
Figure \ref{fig_exp_avgGScanT_compare} shows the time of IP address pairs scanning (GScanT). The graph shows that the GScanT of READ is slightly higher than that of SRLA algorithm. However, the GScanT of each algorithm is not more than 4 seconds, which can process 40 Gb/s of high-speed network traffic in real time.
\begin{figure}[!ht]
\centering
\includegraphics[width=0.47\textwidth]{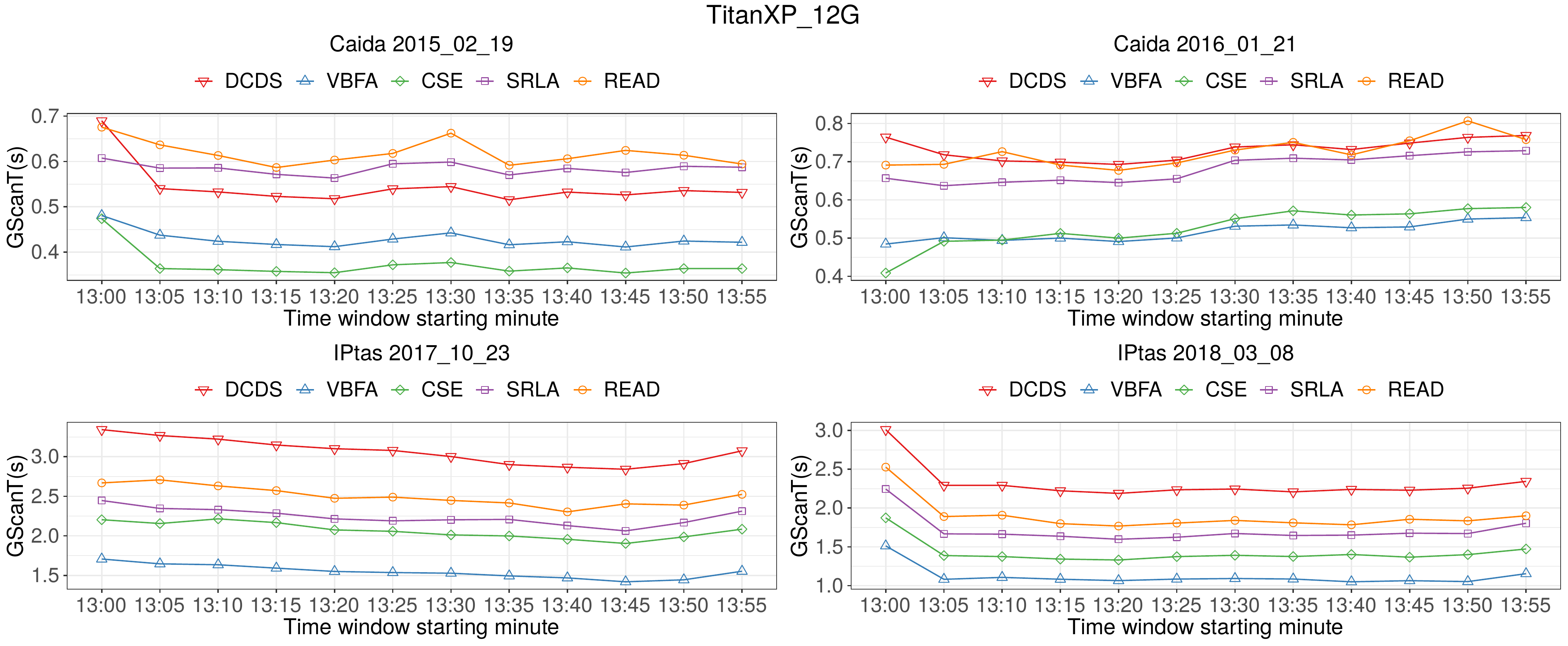}
\caption{Time of scan IP address pair}
\label{fig_exp_avgGScanT_compare}
\end{figure}

Figure \ref{fig_exp_avgGEstT_compare} shows the time of candidate super point cardinality estimation (GEstT). The graph shows that GEstT of READ is close to DCDS, VBFA and SRLA algorithm, much lower than CSE, and GEstT of READ is not higher than 2.5 seconds. Therefore, READ can detect super points in real-time from 40Gb/s high-speed network.

\begin{figure}[!ht]
\centering
\includegraphics[width=0.47\textwidth]{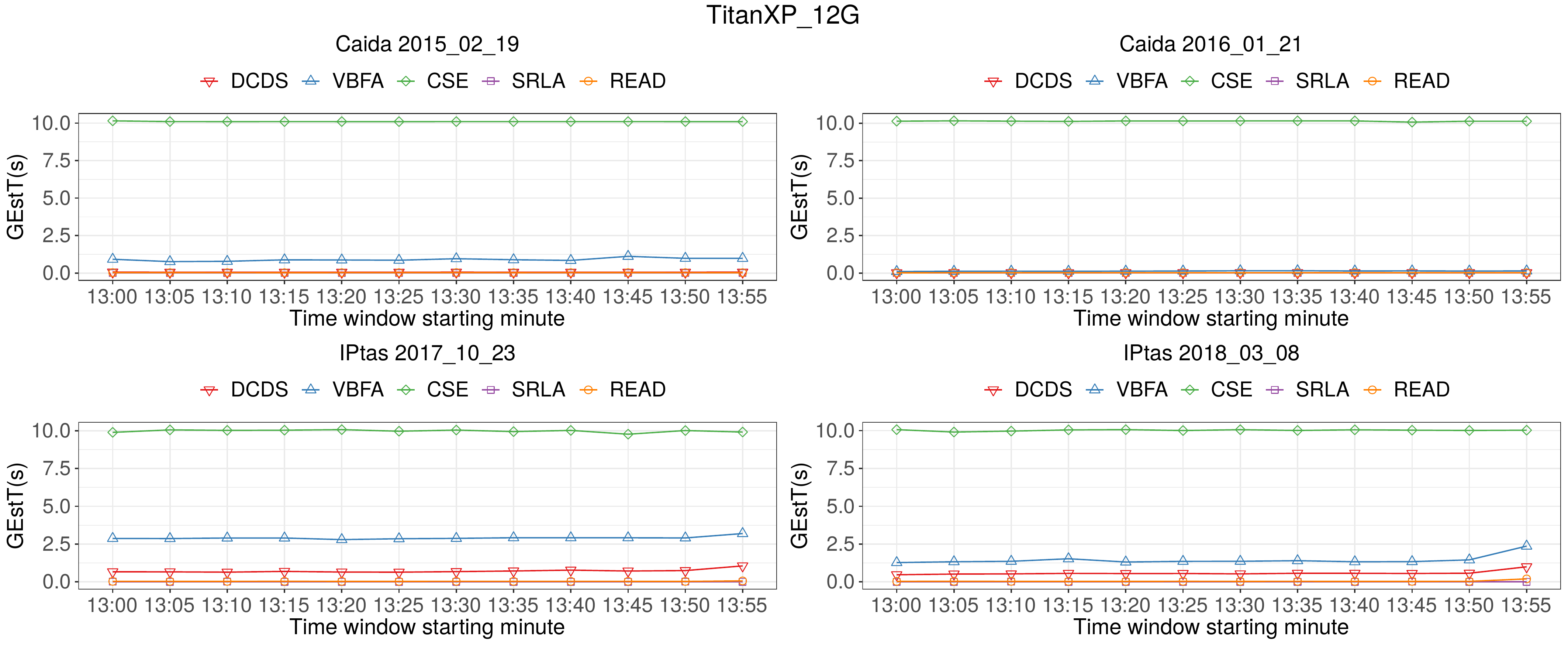}
\caption{Time of estimate candidate super points}
\label{fig_exp_avgGEstT_compare}
\end{figure}

\subsection{Data transmission under distributed environment}
READ is a distributed algorithm. In distributed environment, data will be transmitted between each observation node and the global server, including:
\begin{itemize}
 \item REC from observation node to the global server;
 \item Candidate super points from the global server to each observation node \item The LE set of candidate super points from each observation node to the global server. 
\end{itemize}

In the above data, the size of REC is fixed. The size of candidate super points and LE in transmission depends on the number of candidate super points. From the running process of READ, it can be seen that the candidate super points generated by READ when running in single node environment are the same as those generated when running in distributed environment. Therefore, the number of candidate super points generated at runtime under a single node can be used to determine the size of data transmission between observation nodes and the global server in a distributed environment.

Table \ref{fig_exp_transmit_data_static} lists data transmission between each observation node and the global server. The number of candidate super points is the number of candidate super points produced by REC. The size of candidate super points is multiplied by 4 bytes (each IPv4 address size is 4 bytes); the size of candidate super points' LE is multiplied by $2^{11}$ bytes (LE contains $2^{14}$ bits, $2^{11}$ bytes). The total amount of data transmitted is the sum of the size of REC, the size of candidate super point and the size of LE of candidate super points. The master data structure size is the sum of REC and LEV. The percentage of transmitted data is the ratio of the total amount of transmitted data to the size of the master data structure. From Table \ref{fig_exp_transmit_data_static}, we can see that the average amount of data transmitted by READ between the global server and each observation node is not more than 7.5 MB, which only occupies less than $2.3\%$ of the total size of master data structure.

\begin{table*}
\centering
\caption{Transmitting data between each observation node and the global server}
\label{fig_exp_transmit_data_static}
\begin{tabular}{c}                                                                                                                                                                                                                           
\centering
\includegraphics[width=0.9\textwidth]{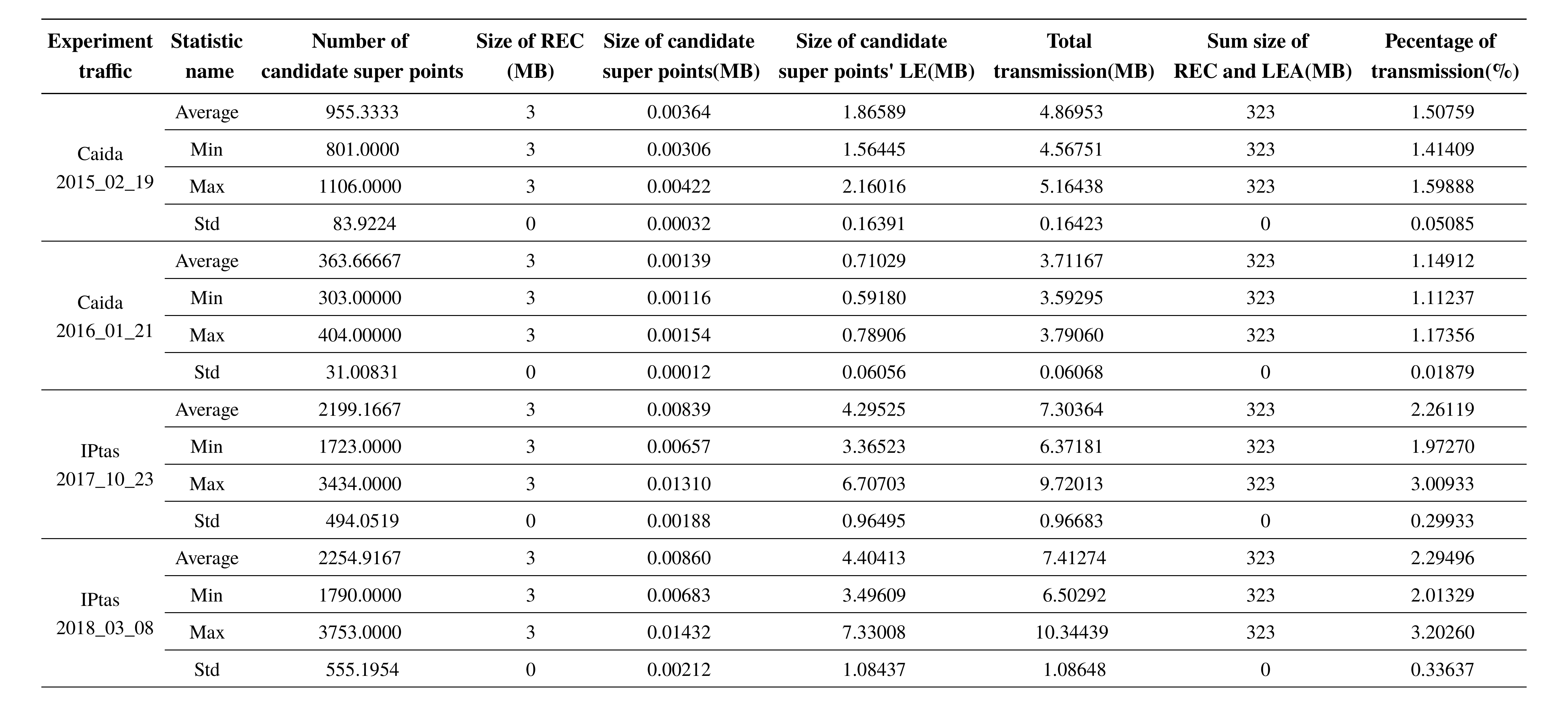}
\end{tabular}
\end{table*}

\subsection{Experiment conclusion}
From the above experiments, the following conclusions can be drawn:
\begin{itemize}
\item The memory consumption and error rate of READ is similar to the existing algorithms.
\item The running time of READ is small enough to handle 40Gb/s networks in real time.
\item In distributed environment, READ only needs to transmit up to 10.4 MB of memory between each observation node and the global server, which accounts for less than $3.21\%$ of the size of master data structure. It is obviously superior to other algorithms and has the advantage of low communication overhead.
\end{itemize}

\section{Conclusion} \label{sec-conclusion}
READ uses REC to generate candidate super points in distributed environment. REC is a three-dimensional structure of RE. Because RE has the characteristics of small memory occupation and fast computing speed, REC can generate candidate super points from 40Gb/s high-speed network with only 3MB of memory. LEA is used to estimate the cardinalities of candidate super points and filter out the super points. READ does not need to transfer the entire LEA to the global server. For 40 Gb/s high-speed network,  the data size transmitted between each observation node and the global server is only $3.21\%$ of the sum of REC and LEA. Low data communication overhead ensures the efficient operation of READ in distributed environment even under the sliding time window.
\section*{Reference}
\iftoggle{ACM}{
\bibliographystyle{ACM-Reference-Format}
}
\iftoggle{IEEEcls}{
\bibliographystyle{IEEEtran}
}
\iftoggle{ElsJ}{
\bibliographystyle{elsarticle-num}
}

\bibliography{..//ref} 

\end{document}

%% file: READ_notation.tex
\usepackage{calrsfs}
\DeclareMathAlphabet{\pazocal}{OMS}{zplm}{m}{n}
\usepackage{amsfonts,amssymb} 
\usepackage{bbm}

\def\ANetmath{\mathbb{A}}
\def\BNetmath{\mathbb{B}}

\def\TWinmath{\mathcal{T}}

\def\ANet{$\ANetmath\ $}
\def\BNet{$\BNetmath\ $}

\def\aipmath{\mathbbm{a}}
\def\bipmath{\mathbbm{b}}

\def\OPmath#1#2{\mathbb{S}^{\BNetmath}_{#1,#2}}

\def\IPLeftPartmath#1{\mathcal{L}_{#1}}
\def\IPLeftPart#1{$\IPLeftPartmath{#1}$}

\def\SubnodeNummath{\mathfrak{n}}

\def\SubnodeIdxmath{\mathfrak{l}}

\def\SubnodeObjmath#1{\mathcal{O}_{#1}}
\def\SubnodeObj#1{$\SubnodeObjmath{#1}$}

\def\IPaironlymath{\mathbb{S}^{pair}}
\def\IPairmath#1#2{\IPaironlymath_{#1,#2}}   
\def\IPair#1#2{$\IPairmath{#1}{#2}$}

\def\RECaryNlgtmath{r}

\def\RECrowNlgtmath{\textit{u}}

\def\RECcolNlgtmath#1{\textit{v}_{#1}}
\def\RECcolNlgt#1{$\RECcolNlgtmath{#1}$}

\def\RECobjonlymath{\mathbb{R}}
\def\RECobjmath#1{\RECobjonlymath^{#1}}
\def\RECobj#1{$\RECobjmath{#1}$}

\def\REidxINithRowmath#1#2{\mathcal{I}_{#1}^{#2}}
\def\REidxINithRow#1#2{$\REidxINithRowmath{#1}{#2}$}

\def\REidxStartBitmath#1{\mathfrak{b}_{#1}}
\def\REidxStartBit#1{$\REidxStartBitmath{#1}$}

\def\REaryObjmath#1{\mathcal{A}_{#1}}
\def\REaryObj#1{$\REaryObjmath{#1}$}

\def\REobjInRECmath#1#2{\RECobjonlymath_{#1}^{#2}}
\def\REcandTupleObjmath#1{\mathfrak{C}_{#1}}
\def\REcandTupleObjmath{\mathfrak{C}}

\def\REcandIdxObjmath#1#2{\mathfrak{c}_{#1}^{#2}}

\def\LEArowNmath{\hat{u}}

\def\LEAcolNmath{\hat{v}}

\def\LEAobjmath{\mathbb{L}}
\def\LEAofOneNodemath#1{\LEAobjmath^{#1}}
\def\LEAofOneNode#1{$\LEAofOneNodemath{#1}$}
\def\LEinLEAofOneNodemath#1#2#3{\LEAofOneNodemath{#3}_{#1,#2}}
\def\LEinLEAofOneNode#1#2#3{$\LEinLEAofOneNodemath{#1}{#2}{#3}$}

\def\LEAhashFunctionNoParmath#1{\mathbbm{h}_{#1}^{LEA}}

\def\LEAhashFunctionmath#1#2{\mathbbm{h}_{#1}^{LEA}(#2)}
\def\LEAhashFunction#1#2{$\LEAhashFunctionmath{#1}{#2}$}

\def\LEophostHashmath#1{\mathbbm{h}^{LE}(#1)}
\def\LEophostHash#1{$\LEophostHash#1$}
\def\REophostrandHashmath#1{\mathbbm{h}^{rand}(#1)}
\def\REophostbitHashmath#1{\mathbbm{h}^{RE}(#1)}

\def\LoadFactormath{\mathfrak{L}}

\def\LEbitsSetmath{\mathbb{C}}

\def\CadIPobjmath{\mathfrak{q}}    
\def\CadIPobjSetmath{\mathbb{Q}}

\def\CIPnumbermath{\mathfrak{w}}

\def\LEobjonlymath{\mathfrak{B}}
\def\LEcipInSLEAmath#1#2#3{\LEobjonlymath_{#1}^{#2}(#3)}
\def\LEcipInGLEAmath#1#2{\LEobjonlymath_{#1}(#2)}
\def\ULEcipInGLEAmath#1{\overline{\overline{\LEobjonlymath(#1)}}}

\def\ULEcipInGLEAREADmath#1{\overline{\LEobjonlymath(#1)}}

\def\ULEcipInSNodemath#1#2{\overline{\LEobjonlymath^{#1}(#2)}}
\def\ULEcipsetInSNodemath#1{\overline{\LEobjonlymath^{#1}}}

\def\ULESetofcipInSNodemath#1{\overline{\LEobjonlymath^{#1}}}
\def\LEusedClusivemath#1{\LEobjonlymath(#1)}

\def\bitInArymath#1#2{\beta_{#1}^{#2}}
\def\bitArraySetmath{\beta}

\def\bipHashedValuemath{\tilde{\bipmath}}

\def\LSBmath#1{\mathcal{R}^{0}(#1)}

\def\greyBG#1{\colorbox[rgb]{0.8,0.8,0.8}{#1}}

\def\HashSetVP{$\HashSetVP\ $}

%% file: READ_table_of_notation_content.tex
\begin{table*}[]
\caption{Notations and symbols used.}
\label{tbl_notationTbl}
\begin{tabular}{ll}
\hline
Notation           & Definition    \\ \hline
\ANet              & The network from which to detect super points.    \\
\BNet              & The network communicating with \ANet through edge routers. \\
$\aipmath$ or $\bipmath$       & An IP address in $\ANetmath$ or $\BNetmath$.        \\
$\TWinmath$        & A time window.\\
$\OPmath{\aipmath}{\TWinmath}$ & Set of opposite hosts of $\aipmath$ in $\TWinmath$.\\
$\SubnodeNummath$ & The number of distributed observation nodes.\\
$\SubnodeObjmath{\SubnodeIdxmath}$ & The $\SubnodeIdxmath$-th observation node.\\
$\IPairmath{\TWinmath}{\SubnodeIdxmath}$ & The stream of IP pair observed on $\SubnodeObjmath{\SubnodeIdxmath}$ in time window $\TWinmath$.\\
$\RECobjmath{\SubnodeIdxmath}$ & A RE cube in the $\SubnodeIdxmath$-th observation node.\\
$\RECaryNlgtmath$ & The number of right bits in $\aipmath$ used to locate a RE array in RE cube. \\
$\IPLeftPartmath{\aipmath}$ & The left $(32-\RECaryNlgtmath)$ bits of $\aipmath$.\\
$\RECrowNlgtmath$ & The number of row in a RE array.\\
$\RECcolNlgtmath{i}$ & The number of bits in $\IPLeftPartmath{\aipmath}$ which is used to locate a RE in the $i$-th row of a RE array.\\
$\LEAofOneNodemath{\SubnodeIdxmath}$ & A LE array in the $\SubnodeIdxmath$-th observation node.\\
$\LEArowNmath$ & The number of row of a LE array.\\
$\LEAcolNmath$ & The number of column of a LE array.\\

\hline

\end{tabular}
\end{table*}

%% file: READ.bbl
\begin{thebibliography}{10}
\expandafter\ifx\csname url\endcsname\relax
  \def\url#1{\texttt{#1}}\fi
\expandafter\ifx\csname urlprefix\endcsname\relax\def\urlprefix{URL }\fi
\expandafter\ifx\csname href\endcsname\relax
  \def\href#1#2{#2} \def\path#1{#1}\fi

\bibitem{report:cnni:chinesenetreport:en}
C.~I. N.~I. Center£¨CNNIC£©,
  \href{http://www.cac.gov.cn/2019-02/28/c_1124175677.htm}{China internet
  network development statistic report(43th)} (Feb. 2019).
\newline\urlprefix\url{http://www.cac.gov.cn/2019-02/28/c_1124175677.htm}

\bibitem{thesis:zap:seu:2015:en}
Z.~Ai-ping, Research on the key issues of traffic measurement in high-speed
  networks, Ph.D. thesis, Southeast University (2015).

\bibitem{ieeec2018:generalidsaccelerationforhighspeednetworks}
J.~{Kucera}, L.~{Kekely}, A.~{Piecek}, J.~{Korenek}, General ids acceleration
  for high-speed networks, in: 2018 IEEE 36th International Conference on
  Computer Design (ICCD), 2018, pp. 366--373.
\newblock \href {http://dx.doi.org/10.1109/ICCD.2018.00062}
  {\path{doi:10.1109/ICCD.2018.00062}}.

\bibitem{hsd:streamingalgorithmfastdetectionsuperspreaders}
S.~Venkataraman, D.~Song, P.~B. Gibbons, A.~Blum, New streaming algorithms for
  fast detection of superspreaders, in: in Proceedings of Network and
  Distributed System Security Symposium (NDSS, 2005, pp. 149--166.

\bibitem{instr:asurveyintrusiondetectiontechniquesincloud:chiragmodi}
C.~Modi, D.~Patel, B.~Borisaniya, H.~Patel, A.~Patel, M.~Rajarajan,
  \href{http://www.sciencedirect.com/science/article/pii/S1084804512001178}{A
  survey of intrusion detection techniques in cloud}, Journal of Network and
  Computer Applications 36~(1) (2013) 42 -- 57.
\newblock \href {http://dx.doi.org/http://doi.org/10.1016/j.jnca.2012.05.003}
  {\path{doi:http://doi.org/10.1016/j.jnca.2012.05.003}}.
\newline\urlprefix\url{http://www.sciencedirect.com/science/article/pii/S1084804512001178}

\bibitem{hsd:infcom:simpleadaptiveidentificationsuperspreadersflowsampling}
N.~Kamiyama, T.~Mori, R.~Kawahara, Simple and adaptive identification of
  superspreaders by flow sampling, in: IEEE INFOCOM 2007 - 26th IEEE
  International Conference on Computer Communications, 2007, pp. 2481--2485.
\newblock \href {http://dx.doi.org/10.1109/INFCOM.2007.305}
  {\path{doi:10.1109/INFCOM.2007.305}}.

\bibitem{hsd:adatastreamingmethodmonitorhostconnectiondegreehighspeed}
P.~Wang, X.~Guan, T.~Qin, Q.~Huang, A data streaming method for monitoring host
  connection degrees of high-speed links, IEEE Transactions on Information
  Forensics and Security 6~(3) (2011) 1086--1098.
\newblock \href {http://dx.doi.org/10.1109/TIFS.2011.2123094}
  {\path{doi:10.1109/TIFS.2011.2123094}}.

\bibitem{hsd:detectionsuperpointsvectorbloomfilter}
W.~Liu, W.~Qu, J.~Gong, K.~Li, Detection of superpoints using a vector bloom
  filter, IEEE Transactions on Information Forensics and Security 11~(3) (2016)
  514--527.
\newblock \href {http://dx.doi.org/10.1109/TIFS.2015.2503269}
  {\path{doi:10.1109/TIFS.2015.2503269}}.

\bibitem{hsd:compactspreadestimatorsmallhighspeedmemory}
M.~Yoon, T.~Li, S.~Chen, J.-K. Peir,
  \href{http://dx.doi.org/10.1109/TNET.2010.2080285}{Fit a compact spread
  estimator in small high-speed memory}, IEEE/ACM Trans. Netw. 19~(5) (2011)
  1253--1264.
\newblock \href {http://dx.doi.org/10.1109/TNET.2010.2080285}
  {\path{doi:10.1109/TNET.2010.2080285}}.
\newline\urlprefix\url{http://dx.doi.org/10.1109/TNET.2010.2080285}

\bibitem{trafficclas2015j:aclassorientedfeatureselection}
Z.~Liu, R.~Wang, M.~Tao, X.~Cai,
  \href{http://www.sciencedirect.com/science/article/pii/S0925231215007870}{A
  class-oriented feature selection approach for multi-class imbalanced network
  traffic datasets based on local and global metrics fusion}, Neurocomputing
  168 (2015) 365 -- 381.
\newblock \href
  {http://dx.doi.org/https://doi.org/10.1016/j.neucom.2015.05.089}
  {\path{doi:https://doi.org/10.1016/j.neucom.2015.05.089}}.
\newline\urlprefix\url{http://www.sciencedirect.com/science/article/pii/S0925231215007870}

\bibitem{tranieee2014:towardsmoreefficientcardinalityestimationforlarge-scalerfidsystems}
Y.~Zheng, M.~Li, Towards more efficient cardinality estimation for large-scale
  rfid systems, IEEE/ACM Transactions on Networking 22~(6) (2014) 1886--1896.
\newblock \href {http://dx.doi.org/10.1109/TNET.2013.2288352}
  {\path{doi:10.1109/TNET.2013.2288352}}.

\bibitem{tranieee2013:contentionbasedestimationneighborcardinality}
H.~Adam, E.~Yanmaz, C.~Bettstetter, Contention-based estimation of neighbor
  cardinality, IEEE Transactions on Mobile Computing 12~(3) (2013) 542--555.
\newblock \href {http://dx.doi.org/10.1109/TMC.2012.19}
  {\path{doi:10.1109/TMC.2012.19}}.

\bibitem{confieee2015:towardsconstanttimecardinalityestimationforlargescalerfidsystems}
B.~Li, Y.~He, W.~Liu, Towards constant-time cardinality estimation for
  large-scale rfid systems, in: 2015 44th International Conference on Parallel
  Processing, 2015, pp. 809--818.
\newblock \href {http://dx.doi.org/10.1109/ICPP.2015.90}
  {\path{doi:10.1109/ICPP.2015.90}}.

\bibitem{dc1983:probabilisticcounting}
P.~Flajolet, G.~N. Martin, Probabilistic counting, in: 24th Annual Symposium on
  Foundations of Computer Science (sfcs 1983), 1983, pp. 76--82.
\newblock \href {http://dx.doi.org/10.1109/SFCS.1983.46}
  {\path{doi:10.1109/SFCS.1983.46}}.

\bibitem{dc:hyperloglogtheanalysisofanearoptimalcardinalityestimationalgorithm}
P.~Flajolet, E.~Fusy, O.~Gandouet, F.~Meunier,
  \href{https://hal.archives-ouvertes.fr/hal-00406166}{{HyperLogLog: the
  analysis of a near-optimal cardinality estimation algorithm}}, in: P.~Jacquet
  (Ed.), {Analysis of Algorithms 2007 (AofA07)}, Juan les pins, France, 2007,
  pp. 127--146.
\newline\urlprefix\url{https://hal.archives-ouvertes.fr/hal-00406166}

\bibitem{dc:alineartimeprobabilisticcountingdatabaseapp}
K.-Y. Whang, B.~T. Vander-Zanden, H.~M. Taylor,
  \href{http://doi.acm.org/10.1145/78922.78925}{A linear-time probabilistic
  counting algorithm for database applications}, ACM Trans. Database Syst.
  15~(2) (1990) 208--229.
\newblock \href {http://dx.doi.org/10.1145/78922.78925}
  {\path{doi:10.1145/78922.78925}}.
\newline\urlprefix\url{http://doi.acm.org/10.1145/78922.78925}

\bibitem{ispa2017:highspeednetworksuperpointsdetectionbasedslidingwindowgpu}
J.~Xu, W.~Ding, J.~Gong, X.~Hu, J.~Liu, High speed network super points
  detection based on sliding time window by gpu, in: 2017 IEEE International
  Symposium on Parallel and Distributed Processing with Applications and 2017
  IEEE International Conference on Ubiquitous Computing and Communications
  (ISPA/IUCC), 2017, pp. 566--573.
\newblock \href {http://dx.doi.org/10.1109/ISPA/IUCC.2017.00092}
  {\path{doi:10.1109/ISPA/IUCC.2017.00092}}.

\bibitem{hpcc2018:srla:arealtimeslidingtimewindowsuperpointcardinalityestimationalgorithmforhighspeednetworkongpu}
J.~{Xu}, W.~{Ding}, J.~{Gong}, X.~{Hu}, S.~{Sun}, \textit{SRLA}: A real time
  sliding time window super point cardinality estimation algorithm for high
  speed network based on gpu, in: 2018 IEEE 20th International Conference on
  High Performance Computing and Communications; IEEE 16th International
  Conference on Smart City; IEEE 4th International Conference on Data Science
  and Systems (HPCC/SmartCity/DSS), 2018, pp. 942--947.
\newblock \href {http://dx.doi.org/10.1109/HPCC/SmartCity/DSS.2018.00156}
  {\path{doi:10.1109/HPCC/SmartCity/DSS.2018.00156}}.

\bibitem{IEEEAccess2019:SRLA}
J.~{Xu}, W.~{Ding}, Q.~{Gong}, X.~{Hu}, H.~{Yu}, A super point detection
  algorithm under sliding time windows based on rough and linear estimators,
  IEEE Access 7 (2019) 43414--43427.
\newblock \href {http://dx.doi.org/10.1109/ACCESS.2019.2908226}
  {\path{doi:10.1109/ACCESS.2019.2908226}}.

\bibitem{j2017:unwisdomcrowdsaccuratelyspottingmaliciousipclustersusinnotsoaccurateipblacklists}
B.~{Coskun}, (un)wisdom of crowds: Accurately spotting malicious ip clusters
  using not-so-accurate ip blacklists, IEEE Transactions on Information
  Forensics and Security 12~(6) (2017) 1406--1417.
\newblock \href {http://dx.doi.org/10.1109/TIFS.2017.2663333}
  {\path{doi:10.1109/TIFS.2017.2663333}}.

\bibitem{j2012:anospfintegratedroutingstrategyqosawareenergysavingipbackbonenetworks}
A.~{Cianfrani}, V.~{Eramo}, M.~{Listanti}, M.~{Polverini}, A.~V. {Vasilakos},
  An ospf-integrated routing strategy for qos-aware energy saving in ip
  backbone networks, IEEE Transactions on Network and Service Management 9~(3)
  (2012) 254--267.
\newblock \href {http://dx.doi.org/10.1109/TNSM.2012.031512.110165}
  {\path{doi:10.1109/TNSM.2012.031512.110165}}.

\bibitem{hsd:linespeedaccuratesuperspreaderidentificationdynamicerrorcompensation}
G.~Cheng, Y.~Tang,
  \href{http://www.sciencedirect.com/science/article/pii/S0140366413001400}{Line
  speed accurate superspreader identification using dynamic error
  compensation}, Computer Communications 36~(13) (2013) 1460 -- 1470.
\newblock \href {http://dx.doi.org/http://doi.org/10.1016/j.comcom.2013.05.006}
  {\path{doi:http://doi.org/10.1016/j.comcom.2013.05.006}}.
\newline\urlprefix\url{http://www.sciencedirect.com/science/article/pii/S0140366413001400}

\bibitem{jnlsd2015:anewrobustchineseremaindertheoremwithimprovedperformanceinfrequencyestimationfromundersampledwaveforms}
L.~Xiao, X.-G. Xia,
  \href{http://www.sciencedirect.com/science/article/pii/S0165168415001954}{A
  new robust chinese remainder theorem with improved performance in frequency
  estimation from undersampled waveforms}, Signal Processing 117 (2015) 242 --
  246.
\newblock \href
  {http://dx.doi.org/https://doi.org/10.1016/j.sigpro.2015.05.017}
  {\path{doi:https://doi.org/10.1016/j.sigpro.2015.05.017}}.
\newline\urlprefix\url{http://www.sciencedirect.com/science/article/pii/S0165168415001954}

\bibitem{bf:anewanalysisoffalsepositiverate}
K.~Christensen, A.~Roginsky, M.~Jimeno,
  \href{http://www.sciencedirect.com/science/article/pii/S0020019010002425}{A
  new analysis of the false positive rate of a bloom filter}, Information
  Processing Letters 110~(21) (2010) 944 -- 949.
\newblock \href {http://dx.doi.org/http://dx.doi.org/10.1016/j.ipl.2010.07.024}
  {\path{doi:http://dx.doi.org/10.1016/j.ipl.2010.07.024}}.
\newline\urlprefix\url{http://www.sciencedirect.com/science/article/pii/S0020019010002425}

\bibitem{VATE2018:sw}
J.~Xu, W.~Ding, X.~Hu, Q.~Gong,
  \href{http://www.sciencedirect.com/science/article/pii/S014036641830625X}{Vate:
  A trade-off between memory and preserving time for high accurate cardinality
  estimation under sliding time window}, Computer Communications 138 (2019) 20
  -- 31.
\newblock \href
  {http://dx.doi.org/https://doi.org/10.1016/j.comcom.2019.02.005}
  {\path{doi:https://doi.org/10.1016/j.comcom.2019.02.005}}.
\newline\urlprefix\url{http://www.sciencedirect.com/science/article/pii/S014036641830625X}

\bibitem{expdata:caida}
C.~for Applied Internet Data~Analysis,
  \href{\url{http://www.caida.org/data/passive}}{The caida anonymized internet
  traces}, online;accessed 2017 (2017).
\newline\urlprefix\url{\url{http://www.caida.org/data/passive}}

\bibitem{expdata:iptracecernetjs:en}
N.~technology key labratory~of Jiangsu Province(Southeast~University),
  \href{\url{http://iptas.edu.cn/src/system.php}}{Ip trace and service
  (iptas)}, http://iptas.edu.cn/src/system.php, Online;accessed 2017 (2017).
\newline\urlprefix\url{\url{http://iptas.edu.cn/src/system.php}}

\end{thebibliography}
